\documentclass[prd,aps,amsfonts,superscriptaddress,nofootinbib,longbibliography,notitlepage,eqsecnum]{revtex4-2}

\usepackage{graphicx}
\usepackage[dvipsnames]{xcolor}
\usepackage{rotating}
\usepackage{amsmath,amssymb,amsthm,isomath}
\usepackage{mathrsfs} 

\usepackage{caption}
\usepackage{subcaption}

\usepackage{physics}
\usepackage{circuitikz}
\usetikzlibrary{arrows.meta} 
\usepackage{bm}
\usepackage{verbatim}
\usepackage{mathtools}
\usepackage{dsfont}

\makeatletter
\pgfcircdeclarebipole{}
    {\ctikzvalof{bipoles/generic/width}} 
    {halfgyrator}                    
    {\ctikzvalof{bipoles/generic/width}} 
    {\ctikzvalof{bipoles/generic/width}} 
    {
        \pgfpathmoveto{\pgfpoint{\pgf@circ@res@left}{0pt}}
        \pgfpathlineto{\pgfpoint{\pgf@circ@res@right}{0pt}}
        \pgfusepath{draw}
    
        \pgfsetlinewidth{2\pgflinewidth}
        
        \pgfpathmoveto{\pgfpoint{\pgf@circ@res@left}{0pt}}
        \pgfpatharc{180}{0}{\pgf@circ@res@right}
        \pgfusepath{draw}
    }

\def\halfgyratorpath#1{\pgf@circ@bipole@path{halfgyrator}{#1}}
\tikzset{half gyrator/.style={\circuitikzbasekey, /tikz/to path=\halfgyratorpath}}
\makeatother

\usepackage[colorlinks=true, urlcolor=violet, linkcolor=blue, citecolor=red, hyperindex=true, linktocpage=true]{hyperref}
\usepackage[capitalise,compress]{cleveref}

\newcommand{\fg}[0]{\mathfrak g}
\newcommand{\fG}[0]{\mathfrak G}

\newtheorem{thm}{Theorem}

\theoremstyle{definition}
\newtheorem{defn}[thm]{Definition}

\makeatletter
\renewcommand{\p@subsection}{}
\renewcommand{\p@subsubsection}{}
\makeatother

\usepackage{mathtools}

\makeatletter
\pgfcircdeclarebipolescaled{instruments}
{
    \anchor{text}{\pgfextracty{\pgf@circ@res@up}{\northeast}
        \pgfpoint{-.5\wd\pgfnodeparttextbox}{
            \dimexpr.5\dp\pgfnodeparttextbox+.5\ht\pgfnodeparttextbox+\pgf@circ@res@up\relax
        }
    }
}
{\ctikzvalof{bipoles/oscope/height}}
{josephson}
{\ctikzvalof{bipoles/oscope/height}}
{\ctikzvalof{bipoles/oscope/width}}
{
    \pgf@circ@setlinewidth{bipoles}{\pgfstartlinewidth}
    \pgfextracty{\pgf@circ@res@up}{\northeast}
    \pgfextractx{\pgf@circ@res@right}{\northeast}
    \pgfextractx{\pgf@circ@res@left}{\southwest}
    \pgfextracty{\pgf@circ@res@down}{\southwest}
    \pgfmathsetlength{\pgf@circ@res@step}{0.25*\pgf@circ@res@up}
    \pgfscope
        \pgfpathrectanglecorners{\pgfpoint{\pgf@circ@res@left}{\pgf@circ@res@down}}{\pgfpoint{\pgf@circ@res@right}{\pgf@circ@res@up}}
        \pgf@circ@draworfill
    \endpgfscope
    \pgfscope
      \pgfpathmoveto{\pgfpoint{\pgf@circ@res@left}{\pgf@circ@res@up}}
      \pgfpathlineto{\pgfpoint{\pgf@circ@res@right}{\pgf@circ@res@down}}
      \pgfpathmoveto{\pgfpoint{\pgf@circ@res@right}{\pgf@circ@res@up}}
      \pgfpathlineto{\pgfpoint{\pgf@circ@res@left}{\pgf@circ@res@down}}
      \pgfusepath{draw}
    \endpgfscope
}
\def\pgf@circ@josephson@path#1{\pgf@circ@bipole@path{josephson}{#1}}
\tikzset{josephson/.style = {\circuitikzbasekey, /tikz/to path=\pgf@circ@josephson@path, l=#1}}

\newcommand{\PP}{\mathbb{P}}
\newcommand{\LL}{\mathcal{L}}

\usepackage{dsfont}

\begin{document}
\author{Ciro Salcedo}
\email{ciro@princeton.edu}
\affiliation{Princeton Quantum Initiative, Princeton University, Princeton, NJ 08544, USA}
\author{Sasha Cocquyt}
\affiliation{Princeton University, Department of Electrical and Computer Engineering, Princeton University, Princeton, NJ 08544, USA}
\author{Andrew A. Houck}
\affiliation{Princeton Quantum Initiative, Princeton University, Princeton, NJ 08544, USA}
\affiliation{Princeton University, Department of Electrical and Computer Engineering, Princeton University, Princeton, NJ 08544, USA}
\author{Andrew Osborne}
\affiliation{Princeton University, Department of Electrical and Computer Engineering, Princeton University, Princeton, NJ 08544, USA}
\title{Gyrators for superconducting circuit design}
\begin{abstract}  

Assuming only the solubility of Kirchhoff's laws, the conservation of energy, and the form of the constitutive relations, we present a flux-charge symmetric Hamiltonian formalism that always admits a symplectic reduction for nonsingular, well-posed circuits involving gyrators, inductors, and capacitors (LCG circuits) with arbitrary non-linearity. The formalism introduces a notion of circuit duality on non-linear non-reciprocal circuits and a family of point transformations that inserts, removes, and moves reciprocal elements across gyrators while preserving the Hamiltonian spectrum. With these transformations, we extend Tellegen's classic result that gyrators make reciprocal elements behave as their dual to non-linear reciprocal elements such as the Josephson junction and quantum phase slip. Furthermore, we present conditions under which non-reciprocal circuits can be transformed to dynamically equivalent reciprocal circuits.

\end{abstract}
\maketitle
\tableofcontents
\section{Introduction}
Superconducting circuits are a leading platform for quantum information processing and simulation \cite{krantz_quantum_2019,Arute2019QuantumSU,kjaergaard_superconducting_2020,devoret_fqi}, necessitating systematic Hamiltonian formalisms describing the dynamics of lumped-element circuits. Historically, circuit theory has guided the discovery of the transmon, fluxonium, GKP, and $0-\pi$ qubit architecture \cite{koch2007,kitaev_encoding,manucharyan_fluxonium_2009,0pi}. Moreover, there are proposals for protected qubits (such as the dualmon \cite{le_doubly_2019}) that require increasingly exotic circuit elements which cannot be readily engineered in near-term devices, such as quantum phase slip junctions. As circuit QED (cQED) architectures are fundamentally designed as lumped-element circuits, a Hamiltonian description of their dynamics directly informs the engineering of qubits and other quantum information processing systems. 

Non-reciprocity, defined as the lack of transmission symmetry when the source and detector are swapped, is fundamental to several superconducting-circuit architectures. The most immediate is qubit readout: the circulators and isolators that protect a qubit from amplifier back-action are intrinsically non-reciprocal devices. However their massive ferrite components occupy substantial real-estate at the mixing-chamber stage of a dilution refrigerator. Using on-chip superconducting non-reciprocal elements, instead, leaves more space for additional qubits or control hardware, but requires a quantization framework for non-reciprocal circuits. Recent work uses Floquet engineering of non-reciprocal couplings to realize high-fidelity parametric gates between detuned qubits \cite{unigate,pargate}, and non-reciprocal elements appear as essential components in proposals for noise-protected qubits \cite{gyratorqubit,grimm_stabilization_2020,stable}. Non-reciprocity also provides a means to break time-reversal symmetry in simulated lattice Hamiltonians, making them a platform for studying chiral topological Bose-Hubbard physics \cite{martinez,roushan_chiral_2017}.
 
 The gyrator is a linear, non-reciprocal, two-port element that couples the voltage across one port to the current through the other. Widely-used non-reciprocal circuit elements, such as the circulator or isolator, can be synthesized from gyrators and reciprocal elements \cite{carlin}. Moreover, gyrators can model the dynamics of a circuit under Floquet drives (often used to effect synthetic magnetic fields for analog simulation) \cite{martinez,KollarNR,Rosen2024}. Finally, gyrators provide a means to construct circuit elements in otherwise inaccessible parameter regimes. Since its proposal, the gyrator has been known to make circuit elements behave like their electromagnetic dual \cite{tellegen}. Thus, a linear inductor can be realized in a physically small space with a gyrator and a capacitor rather than a relatively massive magnetic coil. The quantum phase slip junction, whose direct fabrication remains a challenge, may instead be realized with a Josephson junction and a gyrator.

The historically dominant circuit quantization schemes, the method of nodes and the method of loops, define the Lagrangian of a circuit in terms of a single coordinate (either node fluxes or loop charges, respectively) and recover the conjugate variable by the Legendre transform \cite{vool_introduction_2017}. However, these methods are plagued by pathological circuit topologies that depend on the coordinate choice and obstruct quantization.  Recent work has extended the Burkard-Koch-DiVincenzo (BKD) framework to non-reciprocal elements \cite{pr2019,rymarz_consistent_2022} and addressed singular circuits through perturbative parasitic capacitances \cite{rymarz_consistent_2022,pr_comment}, but these approaches still demand manual reduction for a broad class of pathological topologies \cite{pr2019} and may implicitly assume extended coordinates (defined on $\mathbb{R}$ rather than $S^1$), an assumption violated by the transmon \cite{pr_comment}. The geometric Faddeev--Jackiw (FJ) formalism solves the general quantization problem for conservative non-reciprocal circuits by constructing Hamiltonians from the kernel of a pre-symplectic two-form, sidestepping parasitic regularization entirely \cite{pr2025}. Meanwhile, the flux-charge symmetric formalism of \cite{fluxcharge} guarantees a Hamiltonian for any well-posed LC circuit by expressing potentials in both flux and charge coordinates, but does not extend to non-reciprocal elements. However, its quantization scheme operates at the Lagrangian level offering physical insights into circuits that are obscured by geometric FJ.

In this paper, we extend the flux-charge symmetric formalism of \cite{fluxcharge} to non-reciprocal circuits. We define a flux-charge symmetric Lagrangian for LCG circuits that admits a symplectic reduction for any nonsingular, well-posed circuit, regardless of non-linearity, together with a taxonomy of constraints that reveals their topological origin before any algebra is performed. Unlike the geometric Faddeev--Jackiw scheme of \cite{pr2025}, which reduces a pre-symplectic form before defining a Lagrangian, our formalism encodes the circuit's structural information in the Lagrangian. Through this, our circuit description enables a series of canonical transformations on their Lagrangian. With these, we extend the notion of circuit duality to non-linear LCG circuits and include a partial-dual point transformation that inserts or removes gyrators from a circuit while preserving the Hamiltonian spectrum. With partial duality, we can define continuous isospectral families of circuits and find an extension of Tellegen's element-mimicking result \cite{tellegen} to non-linear elements such as the Josephson junction and quantum phase slip. With this, circuit dynamics in otherwise inaccessible parameter regimes can be realized. Moreover, we provide a set of gyrator reduction rules congruent with our formalism that act as canonical transformations on the circuit Lagrangian. We use these to provide a graph-theoretic classification of reciprocally-simulable circuits via the non-reciprocal coupling graph (NCG), a derived multigraph whose vertices are the really-connected components (defined in Section \ref{sec:notation}) of the circuit and whose edges record the gyrators connecting them.  We prove that any circuit whose NCG is bipartite admits a unitarily equivalent realization by capacitors, inductors and ideal transformers, and that when the NCG is a forest the realization contains no gyrators at all.

\subsection{Assumptions and circuit descriptors}\label{sec:assumptions}

Three distinct pathologies obstruct circuit quantization: rank-deficient kinetic matrices, dynamically-trivial circuits, and ill-posed constitutive relations. Here, we reconcile the terminology used to describe these problems in the past with our work. However, we must first define the scope of our formalism. If a circuit is comprised of only gyrators, inductors, and capacitors, we say that it is an LCG circuit, our central object of study. For concreteness, we count the coherent quantum phase slip and Josephson junctions as capacitors and inductors respectively. Although we do not explicitly include transformers, they appear in our formalism as composite objects from gyrator cascades (Section \ref{sec:combo}). 

We call a circuit well-posed if its equations of motion (EOM) admit a unique physical solution up to gauge freedom. Thus if a circuit admits no physical solutions or admits multiple distinct solutions it is called ill-posed. Non-linear elements such as a QPS that is not in series with an inductor or a JJ without a parallel capacitance generate constitutive relations of the form $\sin(x) = c x$ for some arbitrary constant $c <1$ for some coordinate $x$ that is not constrained elsewhere. The addition of parasitic elements, series inductance or parallel capacitance, can resolve this in a reciprocal circuit. Meanwhile, two voltage sources with distinct voltages in parallel yield equations that admit no solution. We discuss ill-posedness in the presence of gyrators later, in Section \ref{sec:transformations}. There, we define a transformation on circuit Lagrangians that, as a point transformation, must preserve ill-posedness.

Some well-posed circuits may still lack a Hamiltonian description because they lack dynamics entirely; we call these singular circuits. We discuss singularity in the presence of gyrators and expand on the class of singular circuits in Section \ref{sec:combo}. 

Previously, singularity has been used to refer to the kinetic matrices of a given Lagrangian formalism. In particular, the method of nodes and the method of loops use a capacitance or inductance matrix, respectively, whose inverse defines the Legendre transform \cite{vool_introduction_2017}.  If these matrices are rank-deficient, which are caused by pure-capacitor loops in the method of nodes or pure-inductor cuts in the method of loops, then the Hamiltonian cannot be found without a manual choice of the circuit's canonical coordinates. As we discuss below, our flux-charge symmetric Lagrangian encodes a matrix that is analogous to a kinetic matrix and is, in general, rank-deficient. Nevertheless, our formalism admits a Hamiltonian description  as shown in Section \ref{sec:vars}. 

Finally, we discuss units in our scheme. The flux-charge symmetric Lagrangian developed in the following sections treats node fluxes and loop charges on equal footing, and the duality transformation of Section~\ref{sec:transformations} exchanges them outright. To make this exchange dimensionally consistent we adopt the convention that a reference conductance $G_0$ has been set to unity: with $G_0 = 1$, the ratio of flux units to charge units is dimensionless, so flux and charge become numerically interchangeable. All Lagrangians and Hamiltonians in this paper should be read as scalar-valued functions of dimensionless coordinates in this convention. Physical predictions are recovered by restoring $G_0$ in any expression where dimensions must be carried explicitly; for a circuit containing a single gyrator of ratio $G_*$, the choice $G_0 = G_*$ makes that gyrator's partial dual map (Section \ref{sec:rules}) numerically parameter-preserving; with several gyrators, or when a gyration ratio is varied as in Section \ref{sec:family}, $G_0$ is held fixed and each $G_\fg/G_0$ is a dimensionless coupling.
\subsection{Notation}\label{sec:notation}
We first present a brief overview of our notation essential to the main body of this text. A summary table is provided at the end (Table \ref{tab:notation}) for reference. We relegate the majority of formal definitions and proofs to the appendices. 

A circuit is represented as an oriented planar embedded multigraph $G = (\mathcal V, \mathcal E, \mathcal F)$, with vertices $\mathcal V$, edges $\mathcal E$, and loops $\mathcal F$. We assume that each connected subgraph of a circuit's graph is planar so that loops coincide with faces of its component's embedding and so that, following \cite{fluxcharge}, the duality transformation in Section \ref{sec:transformations} is well-defined. As every loop bounds a unique face of its component's embedding, we do not distinguish between faces and loops and we use $\mathcal F$ for both, a prescription justified in \cite{fluxcharge}. For $v $ in $\mathcal V$, we denote $|v\rangle$ to mean a vector in a $|\mathcal V|$ dimensional real vector space satisfying $\langle u | v\rangle = \delta_{uv}$. We adopt similar meanings for $|e\rangle$ and $|l\rangle$ which span the vector spaces $\mathbb{R}^{|\mathcal{E}|}$ and $\mathbb{R}^{|\mathcal{F}|}$, respectively.

The incidence matrix $A$ and orientation matrix $B$ are linear maps on the vector spaces defined above that also satisfy a boundary map condition: 

\begin{subequations}
\label{eqn:boundary_matrices}
\begin{align}
    A_{ev} &=
  \begin{cases}
    +1 & e\text{ directed toward } v \\
    -1 & e\text{ directed away from } v \\
    0  & v \notin e
  \end{cases} \label{eqn:A}\\
     B_{le} &=
  \begin{cases}
    +1 & e, l \text{ oriented alike} \\
    -1 & e, l \text{ oriented unalike} \\
    0  & e \notin l.
  \end{cases}\label{eqn:B}  
\end{align}
\end{subequations}
Equivalently, $A_{ev} = \langle e|A|v\rangle$ and $B_{le} = \langle l|B|e\rangle$. The defining algebraic property of $A$ and $B$ is exactness,
\begin{equation}\label{eqn:exact}
  \ker(B) = \mathrm{im}(A),
\end{equation}
which encodes Kirchhoff's laws as the statement that the composition $BA$ vanishes because the chain complex 
\begin{equation}
    \mathbb R^{|\mathcal V|} \xrightarrow{\,A\,} \mathbb R^{|\mathcal E|} \xrightarrow{\,B\,} \mathbb R^{|\mathcal F|}
\end{equation}
is exact at $\mathbb R^{|\mathcal E|}$ \cite{bamberg_sternberg_v2, branin1966,osborne2023symplectic, fluxcharge}.

We orient circuits so that loops in $\mathcal F$ bound faces uniformly on each component's sphere; see Figure \ref{fig:generic} for an example of the convention. We may further consider a circuit to be constructed from many disjoint connected subgraphs. As will become clear later, our gyrators allow coupling of the dynamics between graph-theoretically disconnected sub-circuits. 
We say that two vertices $u_0,u_1 \in \mathcal{V}$ are \textit{really-connected} if there exists a sequence of vertices $(v_0,v_1,...,v_n)$ such that $v_0 = u_0, \;v_n = u_1$, and for each $1 \leq i \leq n$, the edge $(v_{i-1},v_i) \in \mathcal{E}$ or $(v_i,v_{i-1}) \in \mathcal{E}$. In other words, vertices are really-connected if a path of edges exist between them. This is identical to vertices being "connected" in the graph-theoretic sense; we add this qualifier as gyrators can couple the dynamics between disjoint subgraphs and should not be thought of as "really-connecting" the subgraphs. 

Really-connectedness is an equivalence relation on $\mathcal V$ whose equivalence classes are exactly the $k \geq 1$ disjoint connected components which we will denote with $\mathcal V^{(i)}$ for $i = 1,2,\dots,k$. This partition of the vertices naturally induces a partition of $\mathcal E$: 
\begin{equation}
  \mathcal E^{(i)} = \{e = (u,v) \in \mathcal E \,:\, \{u,v\} \subset \mathcal V^{(i)}\}
\end{equation} 

More generally, $\mathcal X^{(i)}$ with a parenthesized index denotes the restriction of $\mathcal X$ to the $i$-th really-connected component of the graph. Each really-connected component $G^{(i)} = (\mathcal V^{(i)}, \mathcal E^{(i)})$ is regarded as a planar graph embedded on its own sphere $S^2$, and we define $\mathcal F^{(i)}$ to be the face set of that intrinsic embedding, including its own outer face. We take the global loop set to be the disjoint union
\begin{equation}\label{eqn:floopset}
    \mathcal F \;=\; \bigsqcup_{i} \mathcal F^{(i)},
\end{equation}
so that a circuit drawing with $k$ really-connected components carries $k$ outer-face charges rather than one shared outer face. We will maintain our convention of using a superscript to distinguish between different really-connected components. The boundary maps of the connected subgraphs then assemble into the boundary maps of the entire graph,
\begin{subequations}\label{eqn:bmapdecomp}
  \begin{align}
    B &= \bigoplus_{i} B^{(i)} \\
    A &= \bigoplus_i A^{(i)},
  \end{align}
\end{subequations}
with each block exact, $\ker B^{(i)} = \mathrm{im}\, A^{(i)}$, such that $\ker B = \mathrm{im}\, A$. By Euler's formula for a connected planar graph embedded in a sphere, each component satisfies $|\mathcal V^{(i)}| - |\mathcal E^{(i)}| + |\mathcal F^{(i)}| = 2$; each component loop set therefore carries one loop more than its number of independent loops, so that
\begin{equation}\label{eqn:facecount}
    |\mathcal F| = |\mathcal E| - |\mathcal V| + 2k .
\end{equation}
This extra loop per component is what supplies the gauge freedoms counted in Section
\ref{sec:vars} and the duality of Section \ref{sec:transformations}. Explicitly,
\begin{subequations}\label{eqn:gaugekernels}
  \begin{align}
    \ker A^{(i)} &= \mathrm{span}\Big\{\sum_{v \in \mathcal V^{(i)}} \ket v\Big\}\\
    \ker \left(B^{(i)}\right)^\intercal &= \mathrm{span}\Big\{\sum_{l \in \mathcal F^{(i)}} \ket l\Big\},
  \end{align}
\end{subequations}
the first by connectedness of the component and the second because each edge borders two of the component's faces with opposite induced orientations, with no further relation among its loops by (\ref{eqn:facecount}) and rank-nullity. Each kernel is one-dimensional, so the number of null vectors of this sort is precisely twice the number of really-connected components. 

The manifestation of Kirchhoff's rules becomes clearest when we adopt the standard convention that charges and fluxes will be denoted by $q$ and $\phi$ respectively. Particularly, we define these charges and fluxes on topological objects as indicated by its subscript; thus we introduce node fluxes $\phi_v$ and loop charges $q_l$.
We write them as
\begin{subequations}
    \begin{align}
      |\phi\rangle &= \sum_v \phi_v|v\rangle\\
      |q\rangle &= \sum_l q_l|l\rangle
    \end{align}
\end{subequations}
Edge variables are derived from boundary maps acting on these variables: the edge flux on edge $e$ is $\Phi_e = \sum_v A_{ev}\phi_v = \langle e|A|\phi\rangle$, and the edge charge is $Q_e = \sum_l q_l B_{le} = \langle q|B|e\rangle$. Together, $|\phi\rangle$ and $|q\rangle$ parameterize a redundant phase space with more coordinates than a complete description of dynamics requires. We describe the reduction of coordinates to a phase space in Section \ref{sec:vars}. 

As this work restricts to LCG circuits, we partition the set of edges in a circuit diagram according to the form of their constitutive relations:
\begin{equation}
  \mathcal E = \mathcal C \sqcup \mathcal I \sqcup \mathcal G,
\end{equation}
where $\mathcal C$, $\mathcal I$, and $\mathcal G$ are the sets of capacitive, inductive, and gyrative edges respectively. 
We write
$\mathbb P_{\mathcal C}$, $\mathbb P_{\mathcal I}$, $\mathbb P_{\mathcal G}$ for
the orthogonal projectors onto the corresponding edge-space subspaces, so that
$\mathbb P_{\mathcal X} = \sum_{e \in \mathcal X} |e\rangle\langle e|$ for
$\mathcal X \in \{\mathcal C, \mathcal I, \mathcal G\}$ and importantly:
\begin{equation}\label{eqn:projector}
    \mathbb P_{\mathcal C} + \mathbb P_{\mathcal I} + \mathbb P_{\mathcal G}
= \mathbb I,
\end{equation}
the identity matrix. 
We may further partition edges by both their type and by the really-connected component they inhabit:
\begin{subequations}
  \begin{align}
    \mathcal C^{(i)} &= \mathcal C \cap \mathcal E^{(i)} \\ 
    \mathcal I^{(i)} &= \mathcal I \cap \mathcal E^{(i)} \\ 
    \mathcal G^{(i)} &= \mathcal G \cap \mathcal E^{(i)} 
  \end{align}
\end{subequations} 
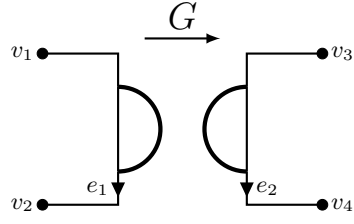
\begin{figure}
  \centering
\begin{circuitikz}[scale=2]
    \draw[thick] (-.5, 1) to (0, 1) to[half gyrator, i_=$e_1$] (0, 0) to (-.5, 0);
    
    \filldraw[black] (-.5, 1) circle (1pt) node[anchor=east]{$v_1$};
    \filldraw[black] (-.5, 0) circle (1pt) node[anchor=east]{$v_2$};
    
    \coordinate (ArrowStart) at (0, 1.1);

\begin{scope}[xshift=10]
    \draw[thick] (1, 1) to (.5, 1) to[half gyrator, mirror, i^=$e_2$] (.5, 0) to (1, 0);
    \filldraw[black] (1, 1) circle (1pt) node[anchor=west]{$v_3$};
    \filldraw[black] (1, 0) circle (1pt) node[anchor=west]{$v_4$};

    \coordinate (ArrowEnd) at (.5, 1.1);
\end{scope}

    \draw[-latex, thick, shorten <=10pt, shorten >=10pt] (ArrowStart) -- node[above,font=\Large] {$G$} (ArrowEnd);

\end{circuitikz}
    \caption{The gyrator symbol proposed in \cite{tellegen}. We include the labeling of vertices and edges to demonstrate the notion that a circuit can be thought of as an object defined on an underlying graph. To reverse the orientation of either edge in the gyrator would flip the sign of the gyration ratio.}\label{fig:symbol}
\end{figure}

Unlike the capacitor and inductor, which are one-port elements, the gyrator is a two-port circuit element coupling the voltage of one port to the current of the other. We therefore regard gyrators graph-theoretically as ordered pairs of edges (Figure \ref{fig:symbol}). The set of gyrator edge pairs of a circuit, which we denote $\fG$, is a collection of ordered edge pairs
\begin{equation}
    \fG \subseteq \mathcal G \times \mathcal G
\end{equation}
that are pairwise disjoint and together cover $\mathcal G = \bigcup_{\mathfrak g \in \fG} \mathfrak g$. 
To each gyrator edge pair $\mathfrak g = (e_1, e_2) \in \fG$ we associate a gyration ratio $G_{\mathfrak g} \in \mathbb R$ with units of conductance; reversing the ordering of an edge pair flips the sign of $G_{\mathfrak g}$, so the edge pairs are taken to be ordered. In conventional notation, the gyrator is additionally labeled by its gyration ratio and an arrow that gives the sign of the constitutive relations. Although we include it in Figure \ref{fig:symbol}, the same information is encoded by the orientation of a gyrator's edges in our notation and thus we exclude it thereafter. Supposing that the gyrator couples the flux and charge through $e_1 = (v_1,v_2)$ to the charge and flux through $e_2 = (v_3,v_4)$, the gyrator has the constitutive relations 
\begin{subequations}\label{eqn:gyrconst}
  \begin{align}
    G_{\mathfrak g}\dot \Phi_{e_1} &=  \dot Q_{e_2} \\
    -G_{\mathfrak g}\dot \Phi_{e_2} &=  \dot Q_{e_1}.
  \end{align}
\end{subequations} 
Only one condition in (\ref{eqn:gyrconst}) is strictly necessary, since the second follows from the first and conservation of energy. Note that the two edges in a gyrator edge pair need not lie in the same really-connected component, thus necessitating the distinction provided by the term.

\begin{table}[h]
  \centering
  \caption{Summary of notation used throughout the paper. }
  \label{tab:notation}
  \begin{tabular}{ll}
    \hline\hline
    Symbol & Meaning \\
    \hline
    $\mathcal V, \mathcal E, \mathcal F$ & Vertices, edges, loops of the circuit graph \\
    $A, B$ & Incidence matrix; orientation matrix \\
    $|v\rangle, |e\rangle, |l\rangle$ & Basis vectors over $\mathbb{R}^{|\mathcal V|}$, $\mathbb{R}^{|\mathcal E|}$, $\mathbb{R}^{|\mathcal F|}$ \\
     \hline
    $\mathcal C, \mathcal I, \mathcal G$ & Capacitive, inductive, and gyrative edges \\
    $\mathbb P_{\mathcal C}, \mathbb P_{\mathcal I}, \mathbb P_{\mathcal G}$ & Projectors onto each edge subset \\
    \hline
    $\mathfrak g$             & A gyrator edge pair \\
    $\fG$             & The set of gyrator edge pairs in a circuit \\
    $G_{\mathfrak g}$ & Gyration ratio of $\mathfrak g$ \\
    \hline
    $\phi_v, q_l$ & Node fluxes; loop charges \\
    $|\phi\rangle, |q\rangle$ & Flux and charge coordinate vectors \\
    $\Phi_e, Q_e$ & Edge flux and edge charge \\
    \hline\hline
  \end{tabular}
\end{table}

\subsection{Brief review of flux-charge symmetric circuit mechanics}\label{sec:review}
Lagrangian systems admit a simple treatment of constraints via Lagrange multipliers, while Hamiltonian descriptions require Dirac brackets to do the same. Lumped-element circuit models generally have many nontrivial constraints arising from either Kirchhoff's laws or constitutive relations manifesting as an excess of coordinate freedom. As such, it is not clear how to produce a Hamiltonian directly from a circuit diagram. Instead we begin with a Lagrangian description and integrate out constraints, incrementally paring down configuration space into a suitable phase space in what amounts to a Hamiltonian principle of least action.

The flux-charge symmetric Lagrangian for reciprocal LC circuit \cite{fluxcharge} is given by 
\begin{equation}\label{eqn:fc}
  \mathscr L = \sum_{l,v} q_l M_{lv} \dot\phi_v
  - E
\end{equation}
where 
\begin{equation}\label{eqn:connection}
  M = \tfrac{1}{2} B(\mathbb P_{\mathcal C} - \mathbb P_{\mathcal I}) A,
\end{equation}
is the so-called connection matrix. 
The first term of (\ref{eqn:fc}) is called the symplectic term, and $E$ is the sum of energies stored in circuit elements written entirely in terms of $\phi_v$ and $q_l$ variables, which we will call the energetic term. 

The principle of least action gives 
\begin{subequations}\label{eqn:lac}
\begin{align}
    0 &= \sum_v M_{lv} \dot \phi_v  - \frac{\partial E}{\partial q_l}, \\ 
    0 &= -\sum_l   \dot q_l M_{lv} - \frac{\partial E}{\partial \phi_v}.
\end{align}    
\end{subequations}
Supposing for the moment that the connection matrix is invertible, we define 
\begin{equation}
\{f,g\} = \sum_{l,v} (M^{-1})_{vl}\left( \frac{\partial f}{\partial \phi_v}\frac{\partial g}{\partial q_l} - \frac{\partial f}{\partial q_l}\frac{\partial g}{\partial \phi_v}\right),
\end{equation}
for any function $f$ and $g$.
Hence, (\ref{eqn:lac}) can be rewritten succinctly as $\dot f = \{f, E\}$ for any function $f$ of configuration variables. 
The procedure of inverting $M$ in order to define a Poisson bracket is precisely the Legendre transform.
In reality $M$ fails to be invertible, so we must identify vectors $c_v$  such that $\sum_{l,v} M_{lv} c_v = 0$, which give rise to algebraic constraints of the form 
\begin{equation}
    \sum_{v } c_v \frac{\partial E}{\partial \phi_v} = 0.
\end{equation}
Solving such a constraint and rewriting the Lagrangian on-shell reduces the number of coordinates in the symplectic term which must nonetheless remain linear, effectively reducing the size of the connection matrix without reducing its rank. Repeating this procedure for all such null vectors of $M$ yields a symplectic term with a genuinely invertible connection matrix, allowing a straightforward Legendre transform and subsequent canonical quantization.
In the flux-charge symmetric Lagrangian formalism, there is one such algebraic constraint for each independent purely capacitive loop, purely capacitive cut, purely inductive loop, and purely inductive cut, together with the gyrator-mixing constraints of Section \ref{sec:vars}. The assumption of well-posedness is equivalent to the assertion that all such constraints admit unique solutions. 

The energetic term of the Lagrangian can be further decomposed as 
\begin{equation}
    E = \sum_{e \in \mathcal C} E_e^C\left( \langle q | B|e \rangle \right) + \sum_{e \in \mathcal I} E_e^I\left(\langle e|A |\phi\rangle\right)
\end{equation}
where $E_e^C$ ($E_e^I$) is the energy stored in the capacitor (inductor) on edge $e$.
Importantly, the flux-charge symmetric Lagrangian formalism admits a straight-forward notion of circuit duality. 
If $\mathscr L(q_l, \phi_v)$ is the flux-charge symmetric Lagrangian of some circuit, then $\mathscr{L}(\phi_v, q_l)$ is the flux-charge symmetric Lagrangian of its dual. 

\section{Flux--charge symmetric description of gyrators}\label{sec:flux_charge_LCG}

We seek a generic procedure for writing Lagrangians with gyrators that produces the constitutive relations while satisfying Kirchhoff's laws for nonsingular well-posed circuits. 
For every gyrator edge pair $\mathfrak{g} \in \mathfrak{G}$, we introduce the following term to represent its contribution to the Lagrangian:
\begin{equation}\label{eqn:gyrterm}
  \mathscr{L}_\mathfrak{g} = \frac{G_\mathfrak{g}}{2}\sum_{u,v}\phi_u A_{e_1u}A_{e_2v}\dot\phi_v - \frac{1}{2G_\mathfrak{g}}\sum_{l,l'} q_l B_{le_1}B_{l'e_2}\dot q_{l'}.
\end{equation}
For a gyrator edge pair $\fg = (e_1, e_2)$ with ratio $G_\fg$, we set $G_{(e_2,e_1)} := -G_{(e_1,e_2)}$. Then $\LL_\fg$ is unchanged under reversal of the pair ordering and reversal of the orientation of either gyrative edge accompanied by $G_\fg \to -G_\fg$. 

We may alternatively write the contribution of $\LL_\fg$ in ket-notation, equal to Equation \ref{eqn:gyrterm} up to a total time derivative, to make the properties of the boundary maps more explicit. For $\mathfrak g = (e_1, e_2) \in \fG$, define
\begin{equation}\label{eqn:gdefn}
  \Gamma_{\mathfrak g} = |e_1\rangle\langle e_2| - |e_2\rangle\langle e_1|,
\end{equation}
which we will refer to as the gyrator's antisymmetric form. We then may write the contribution of a single gyrator to a Lagrangian as:
\begin{equation}
  \mathscr{L}_\mathfrak{g} = \frac{1}{4} \left(G_\mathfrak{g}\langle \phi| A^\intercal \Gamma_\mathfrak{g} A |\dot\phi\rangle - G_\mathfrak{g}^{-1}\langle q| B \Gamma_\mathfrak{g} B^\intercal |\dot q\rangle \right)
\end{equation} 
Direct computation gives $\Gamma_{\mathfrak g}^2 = -(|e_1\rangle\langle e_1| + |e_2\rangle\langle e_2|)$, and since distinct gyrator edge pairs are disjoint, $\Gamma_{\mathfrak g}\Gamma_{\mathfrak g'} = 0$ whenever $\mathfrak g \neq \mathfrak g'$. Summing over $\fG$ yields the identity
\begin{equation}\label{eqn:gammaidents}
  \sum_{\mathfrak g \in \fG} \Gamma_{\mathfrak g}^2 = -\mathbb P_{\mathcal G}
\end{equation}
which we invoke throughout Section \ref{sec:flux_charge_LCG}. 
 
We will study the circuit Lagrangian 
\begin{equation}\label{eqn:no_omega_lagrangian}
  \mathscr L = \langle q | M | \dot \phi\rangle + \sum_{\mathfrak{g} \in \fG} \mathscr{L}_\mathfrak{g} - \sum_{e \in \mathcal C} E_e^C - \sum_{e \in \mathcal I} E_e^I
\end{equation} 
where $E_e^C$ is the energy associated with a capacitor on edge $e$ and must depend only upon $Q_e$, and likewise  $E_e^I $ is the energy stored in an inductor on edge $e$ and may depend only upon $\Phi_e $. 

We have chosen $\mathscr L_{\mathfrak g}$ as notation to emphasize the fact that gyrators do not store energy. Moreover, we remark that the symplectic term together with the contributions to the Lagrangian from gyrators admits an uncomplicated block structure reminiscent of (\ref{eqn:fc}). Namely, 
\begin{equation}
\langle q | M |\dot\phi\rangle + \sum_{\mathfrak g \in \fG} \mathscr L _{\mathfrak g} =
\frac{1}{4}
\begin{pmatrix} \langle \phi | & \langle q | \end{pmatrix}
\begin{pmatrix}
    \sum_{\mathfrak{g} \in \mathfrak{G}} G_\mathfrak{g}A^\intercal \Gamma_\mathfrak{g}A & -2M^\intercal \\
    2M & -\sum_{\mathfrak{g} \in \mathfrak{G}} G_\mathfrak{g}^{-1}B\Gamma_\mathfrak{g}B^\intercal
\end{pmatrix}
\begin{pmatrix} |\dot \phi \rangle \\ | \dot q \rangle \end{pmatrix},
\end{equation}
which is clearly antisymmetric.

There are a pair of block decompositions that will be valuable to us: 
firstly, 
\begin{equation}\label{eqn:omega}
\Omega = 
    \begin{pmatrix}
        \sum_{\mathfrak{g} \in \mathfrak{G}} G_\mathfrak{g}A^\intercal \Gamma_\mathfrak{g}A & -2M^\intercal \\ 2M & -\sum_{\mathfrak{g} \in \mathfrak{G}} G_\mathfrak{g}^{-1}B\Gamma_\mathfrak{g}B^\intercal
    \end{pmatrix}
\end{equation}
is antisymmetric and immediately calls to mind \cite{faddeevjackiw} and the technology that can be invoked therein. 
Secondly, we may write
\begin{equation}\label{eqn:upsilon}
    \sigma \Omega =    ( B \oplus A^\intercal) \Theta (A \oplus B^\intercal)
\end{equation}
where $\sigma: \mathbb R^{|\mathcal V|}\oplus\mathbb R^{|\mathcal F|} \to \mathbb R^{|\mathcal F|}\oplus\mathbb R^{|\mathcal V|}$, $(x,y)\mapsto(y,x)$, is the block transposition, and with
\begin{equation}\label{eqn:Thetablock}
    \Theta = 
    \begin{pmatrix}
         \mathbb{P}_\mathcal{C} - \mathbb{P}_\mathcal{I} & -\sum_{\mathfrak{g}\in \mathfrak{G}} G_\mathfrak{g}^{-1} \Gamma_\mathfrak{g} \\ 
         \sum_{\mathfrak{g}\in \mathfrak{G}} G_\mathfrak{g} \Gamma_\mathfrak{g} & \mathbb{P}_\mathcal{I} - \mathbb{P}_\mathcal{C} \\ 
     \end{pmatrix}.
\end{equation}
Of note, $\Theta$ is an involution satisfying $\Theta^2 = 1$ and hence $\sigma \Omega$ is the involution $\Theta$ conjugated by a pair of maps satisfying an exactness condition. As $\sigma$ is an isometry, $\ker\Omega = \ker\sigma\Omega$.

\subsection{Correctness and coordinate choices}\label{sec:coor}

A Lagrangian is ``correct'' if and only if it produces the correct EOM on configuration space. In our case, these EOM are precisely Kirchhoff's laws and the constitutive relations of circuit elements. 
We begin by showing that the circuit constitutive relations follow. 

Take (\ref{eqn:no_omega_lagrangian}), and observe that, for a particular loop $l$, 
\begin{equation}\label{eqn:capconst}
    \frac{\delta}{\delta q_l} \int \mathrm d t \mathscr L = \langle l | M | \dot \phi \rangle +  \frac{\delta}{\delta q_l}\int\mathrm dt \sum_{\fg \in \mathfrak{G}}\mathscr L_\fg  - \sum_{e \in \mathcal C} \frac{\partial E_e^{C}}{\partial q_l}
\end{equation}
Rewriting $M = B(\mathbb P_{\mathcal C} + \frac{1}{2}\mathbb P_{\mathcal G} )A$, we find that 
\begin{equation}\label{eqn:vecconst}
     \frac{\delta}{\delta q_l} \int \mathrm d t \mathscr L = 
     \sum_{e \in \mathcal C, v} B_{le} \left[ 
       A_{ev} \dot\phi_v - (E_e^C)' 
     \right]
     + 
     \frac{1}{2}\sum_{e \in \mathcal G}B_{le}\left[
        A_{ev}\dot\phi_v - G_{(e,e')}^{-1} \sum_{l'}B_{l'e'}\dot q_{l'}
        \right],
\end{equation}
where $(E_e^C)'$ denotes the derivative of $E_e^C$ with respect to its argument $\sum_l q_l B_{le}$ and $e'$ is the partner in a gyrator edge pair to $e$ in $\mathcal G$. 
We recognize that, since (\ref{eqn:vecconst}) holds on-shell for any $l$, it specifies a right null vector of $B$.  
By the exactness of $B$ and $A$, it follows that each term in the sum above vanishes individually up to the trivial exception that there may exist nodes in a circuit (such as a node connecting two capacitors in series) whose flux does not appear in $\mathscr L$ and may satisfy (\ref{eqn:vecconst}) at any value. See \cite{fluxcharge,osborne2024stochastic} as well as Section \ref{sec:vars} for more discussion. 
The other constitutive relations follow analogously by rearranging $\frac{\delta S}{\delta \phi_v}$ in place of $\frac{\delta S}{\delta q_l}$, demonstrating the validity of (\ref{eqn:gyrterm}) and its agreement with \cite{pr2025}. 

On the other hand, writing $M = -B(\mathbb P_{\mathcal I} + \frac{1}{2} \mathbb P_{\mathcal G})A$ instead yields 
\begin{equation}
    0 =-\sum_{e \in \mathcal I} B_{le} A_{ev}\dot\phi_v -\sum_{e \in \mathcal C} B_{le} (E_e^C)' 
     -\frac{1}{2}\sum_{e \in \mathcal G}B_{le}\left[
        A_{ev}\dot\phi_v + G_{(e,e')}^{-1} \sum_{l'}B_{l'e'}\dot q_{l'}
        \right].
\end{equation}
The first term is transparently the voltage differences across inductors along the loop $l$. The second and third terms may also be interpreted as voltage differences across capacitors and gyrators respectively, now that we have established the imposition of their constitutive relations. Analogously, $\frac{\delta S}{\delta \phi_v}$ imposes conservation of current at a vertex. 

Of course, there are many correct Lagrangians to describe a given electrical circuit. Trivially, the addition of a total time derivative leaves dynamics unchanged. We are also free to introduce and integrate out Lagrange multipliers. For concreteness write $\mathscr L$ as in (\ref{eqn:no_omega_lagrangian}), and consider 
\begin{equation}
    \mathscr L' = \mathscr L + \sum_{e \in \mathcal C \cup \mathcal G} \lambda_e \left(\sum_l q_l B_{l e} - Q_e\right).
\end{equation}
Since $\mathscr L$ has no dependence upon $Q_e$, this change of variables does not affect dynamics and gives rise to only a point transformation. After integrating out $\lambda_e$, we find 
\begin{equation}
    \mathscr L' = \left(\sum_{e \in \mathcal C,v}+
    \frac{1}{2} \sum_{e \in \mathcal G}\right)
    Q_e A_{ev} \dot\phi_v 
    + \sum_{\fg = (e,e')}\Big(\frac{G_{\fg}}{2}\sum_{u,v}  \phi_u A_{eu}A_{e'v}\dot\phi_v  -  \frac{1}{2G_{\fg}}\, Q_e \dot Q_{e'} \Big)
    - \sum_{e \in \mathcal C} E_e^C(Q_e) - \sum_{e \in \mathcal I} E_e^I\left(\sum_v A_{ev}\phi_v\right).
\end{equation}
Since dynamical variables live on edges and vertices in these coordinates, we say that these are edge-vertex coordinates. It is also possible to define face-edge coordinates via the Lagrange multiplier 
\begin{equation}
    \mathscr L' = \mathscr L + \sum_{e \in \mathcal I \cup \mathcal G} \lambda_e \left(\Phi_e- \sum_v A_{ev}\phi_v\right).
\end{equation}

Since the only Lagrangian terms coupling distinct really-connected components are gyrator terms, we are free to choose coordinates on each really-connected component independently. For the remainder of this paper, we will not belabor this Lagrange multiplier manipulation and only remark on it when technically necessary. We emphasize that one may just as easily describe circuit Lagrangians in terms charges on vertices and fluxes on loops which preserve the dynamical description of the circuit despite changing the form of the EOM.

\subsection{Constraints and Quantization}\label{sec:vars}

The problems of constraint resolution and circuit quantization are equivalent provided that the reader is accepting of canonical quantization. 
In the present setting, the question of quantizability has been resolved by \cite{pr2025}, but we nonetheless include some discussion of constraints here. 
In order to manipulate Lagrange's equation into a form compatible with Hamiltonian dynamics, $\Omega$ must be invertible, but it is not. 
Hence, there must arise constraints revealing relationships between overcomplete degrees of freedom, which must in turn be solved in order to reduce the size of $\Omega$ without decreasing its rank. 

To illustrate this point, notice that for any linear combination of functional derivatives 
\begin{equation}
  \frac{\delta}{\delta X} := \sum_v c_v \frac{\delta }{\delta \phi_v} + \sum_l d_l \frac{\delta }{\delta q_l},
\end{equation} 
the principle of least action enforces that, on-shell,
\begin{equation}\label{eqn:plac}
  \frac{\delta }{\delta X} \int \mathrm d t \left(\langle q | M | \dot\phi \rangle + \sum_{\fg \in \fG} \mathscr L_{\fg}\right) = \frac{\delta }{\delta X} \int \mathrm d t \, E = \sum_v c_v \frac{\partial E}{\partial \phi_v} + \sum_l d_l \frac{\partial E}{\partial q_l}.
\end{equation} 
If the left hand side of (\ref{eqn:plac}) vanishes identically, the resulting equation of motion  is called a \emph{null vector constraint} and is purely algebraic in the sense that it does not depend on any time differentiated variable. A central theme of \cite{osborne2023symplectic,fluxcharge} is that there is one such algebraic constraint for every null vector of $\Omega$. 
Effectively, null vectors of $\Omega$ encode Lagrange multipliers inherited from Kirchhoff's laws and constitutive relations alone, which we can use to reduce the total number of variables appearing in $\mathscr L$. 
Although it is always possible to find the null vectors of a matrix numerically, for example, the solubility of the resulting algebraic constraint is not immediate. In particular the assumption of well-posedness is equivalent to the assumption that algebraic constraints admit unique solutions. Solubility aside, null vector constraints always encode some circuit simplification that can be carried out (at least in principle) on the diagram itself and hence we nonetheless discuss the details of what kinds of null vector constraints can arise. 

From the block structure of $\Omega$ in (\ref{eqn:omega}), it is clear that the kernel of $\Omega$ decomposes into two structurally distinct sources. Any $\ket{n} \in \ker(A\oplus B^\intercal )$ satisfies $\Omega\ket{n} = 0$ trivially: these are gauge directions inherited from the boundary maps themselves, namely the right null vectors of $A^{(i)}$ and $(B^{(i)})^\intercal$ for each really-connected component.
These null vectors correspond to the freedom to shift a global flux on each really-connected component's vertex set and a global charge on its loop set. By Equation \ref{eqn:gaugekernels}, each really-connected component contributes one null vector of each kind. The remaining null vectors have genuine dynamical content and are best dissected by examining $\Omega$ in the form of Equation \ref{eqn:upsilon}, where their structure is governed entirely by the central matrix $\Theta$. First we will consider arbitrary partitions of loops and vertices as $L \subseteq \mathcal{F}$ and $U \subseteq \mathcal{V}$ respectively. Additionally, we write:
\begin{subequations}
    \begin{align}
          \ket L &= \sum_{l \in L \subset \mathcal F} \alpha_l B^\intercal \ket{l}\\
          \ket U &= \sum_{v \in U \in V}\alpha_v A\ket v
    \end{align}
\end{subequations}
for $\alpha_i \in \mathbb{R}$. Every $\ket x = (\ket\phi, \ket q)^\intercal$ determines a pair of edge vectors $\ket\lambda = (A\ket\phi,\, B^\intercal\ket q)^\intercal = (\ket U, \ket L)^\intercal$, and by (\ref{eqn:gaugekernels}) the map $\ket x \mapsto \ket\lambda$ has a$2k$-dimensional kernel, namely the gauge directions. From (\ref{eqn:upsilon}), $\sigma\Omega\ket x = (B \oplus A^\intercal)\, \Theta\ket\lambda$, so $\ket x$ lies in $\ker\sigma\Omega = \ker\Omega$ precisely when $\Theta \ket\lambda \in \ker(B \oplus A^\intercal)$. By exactness that kernel is $\mathrm{im}(A \oplus B^\intercal)$, so the condition is that $\Theta\ket\lambda$ is again a pair of the same form:
\begin{equation}\label{eqn:kerchar}
    \Theta \begin{pmatrix} \ket U \\ \ket L \end{pmatrix}
    = \begin{pmatrix} A\ket{\phi'} \\ B^\intercal \ket{q'} \end{pmatrix}
    \quad\text{for some } \ket{\phi'},\, \ket{q'} .
\end{equation}
The dynamical null vectors are therefore exactly those $\ket\lambda = (\ket U, \ket L)^\intercal$ that $\Theta$ returns to that form, and $\dim\ker\Omega$ exceeds the number of such independent $\ket\lambda$ by the $2k$ gauge directions. Because $\Theta^2 = 1$, $\Theta$ maps this collection of $\ket\lambda$ bijectively onto itself, and an involution on a space splits it into its $+1$ and $-1$ eigenspaces; the collection therefore has a basis of vectors satisfying
\begin{equation}\label{eqn:Thetaeigen}
    \Theta\ket\lambda = \pm\ket\lambda .
\end{equation}
We classify the dynamical null vectors through such a basis. Note that a general null vector is a sum of a $+1$ and a $-1$ eigenvector and need not itself be an eigenvector, so (\ref{eqn:Thetaeigen}) characterizes a spanning set rather than every null vector.

Writing $\Theta$ as the block matrix from \eqref{eqn:Thetablock} and expanding \eqref{eqn:Thetaeigen}  gives the pair of conditions
\begin{subequations}\label{eqn:Thetaeigen_expanded}
    \begin{align}
    (\mathbb P_{\mathcal C} - \mathbb P_{\mathcal I})\ket{U} - \sum_{\mathfrak g \in \fG} G_{\mathfrak g}^{-1}\,\Gamma_{\mathfrak g}\,\ket{L} &= \pm  \ket{U}\\
        \sum_{\mathfrak g \in \fG} G_{\mathfrak g}\,\Gamma_{\mathfrak g}\,\ket{U} + (\mathbb P_{\mathcal I} - \mathbb P_{\mathcal C})\ket{L} &= \pm \ket{L}.
    \end{align}
\end{subequations}

Recall that $\ket{\lambda} \in \mathbb{R}^{2|\mathcal{E}|}$ and exist on a space of vectors representing edges (circuit elements). Our first structurally distinct type of $\Theta$ eigenvector then corresponds to cuts or loops composed of only a single element class which we call the \textit{pure null vectors}. The pure null vectors are supported on exactly one of the coordinates block such that they take the form $\ket\lambda = (\ket U,0)^\intercal$ or $\ket\lambda = (0, \ket L)^\intercal$. These pure null vectors correspond, in agreement with the formalism of \cite{fluxcharge}, to cuts and loops composed of only a single element class. They are identifiable directly from the circuit's topology without algebraic computation. We will examine specifically the case for a null vector from a pure-cut as the pure-loop null vectors are easily found by the same argument. Should $\ket{l} =0$ (thus, $\ket{L}= 0$ too), then the conditions become:
\begin{subequations}\label{eqn:pureeig}
    \begin{align}
      (\mathbb P_{\mathcal C} - \mathbb P_{\mathcal I})\ket{U} &= \pm  \ket{U}\\
        \sum_{\mathfrak g \in \fG} G_{\mathfrak g}\,\Gamma_{\mathfrak g}\,\ket{U}  &= 0.
    \end{align}
\end{subequations}
We specify that $\ket{L}= 0$ to emphasize that we are considering a vector that is not in $\ker B^\intercal$, but rather a null vector of $\Omega$ that corresponds to a functional derivative without any charge coordinates.

The two equations are satisfied whenever $\ket U$ has support on a single type of reciprocal element (purely capacitive cut or purely inductive cut). Similarly, should a loop null vector $\ket L$ be made of loops of only capacitors or only inductors, we find a null vector and therefore a constraint.

The remaining null vectors are \emph{gyrator-mixing}: linear combinations of basis vectors from both coordinate blocks, so $\ket \lambda  = ( \ket U,\ket L)^\intercal$ with both $\ket U$ and $\ket L$ nonzero. Projecting Equation \ref{eqn:Thetaeigen_expanded} onto the space of gyrative edges with $\mathbb{P}_{\mathcal{G}}$ isolates the genuinely block-mixing content of the constraint:
\begin{subequations}\label{eqn:gyrmix}
       \begin{align}
    - \sum_{\mathfrak g \in \fG} G_{\mathfrak g}^{-1}\,\Gamma_{\mathfrak g}\,\ket{L} &= \pm \mathbb{P}_{\mathcal{G}} \ket{U}\\
        \sum_{\mathfrak g \in \fG} G_{\mathfrak g}\,\Gamma_{\mathfrak g}\,\ket{U} &= \pm \mathbb{P}_{\mathcal{G}}\ket{L}.
    \end{align}
\end{subequations}
each of which couples $\ket U$ and $\ket L$ through gyrator edges via $\Gamma_{\mathfrak g}$. These gyrator-mixing constraints, unlike the pure constraints, are not readily identifiable from loops and cuts on some partition of edges. Instead, recognizing these constraints requires computation. However, finding all the null vectors of $\Omega$ may always be done by polynomial-time matrix inversion.
Here we note that this reduction may always be carried out. The variables removed by the null-vector constraints are precisely the directions along which $\Omega$ vanishes, so the symplectic term $\tfrac14\langle x|\Omega|\dot x\rangle$ does not contain them: they enter $\LL$ only through the energy, undifferentiated. Their Euler-Lagrange equations are therefore algebraic, containing no time derivatives. A well-posed circuit is by definition one in which those equations have a unique solution. Substituting it back changes only the energy and leaves the symplectic term intact, regardless of nonlinearity in the element energies. The antisymmetry of $\Omega$ then guarantees that what remains is a well-defined phase space. The eigenvalues of any real antisymmetric matrix are purely imaginary, and the nonzero eigenvalues occur in pairs $\pm\mathrm i\lambda$ with $\lambda > 0$, so $\mathrm{rank}(\Omega) = 2\kappa$ for a positive integer $\kappa$. The case $\kappa = 0$ requires $\Omega$ to be the zero matrix, corresponding to a singular circuit with no dynamics. Therefore $\tilde\Omega$, the restriction of $\Omega$ to $\mathrm{im}(\Omega)$, is a $2\kappa \times 2\kappa$ antisymmetric full-rank matrix corresponding to a non-degenerate symplectic form on the phase space $\mathbb R^{|\mathcal V| + |\mathcal F|} / \ker(\Omega)$. Our formalism therefore always admits a Hamiltonian, and the symplectic reduction can be performed for any nonsingular, well-posed LCG circuit.

It would be desirable to produce a finite list of patterns that could be recognized from a circuit drawing and each herald some null vector constraint while perhaps providing instructions on exactly which functional derivative must be taken to produce the relevant equality. In the case of reciprocal circuits, as mentioned above, constraints come only from pure null vectors which may be associated with loops or cuts of only a single kind of circuit element. In the case of nonreciprocal circuits, however, the list of patterns that must be recognized fails to be finite (a gyrator-mixing null vector may involve an arbitrarily long chain of gyrators, each contributing a coupled pair of conditions in Equation \ref{eqn:gyrmix}). As such, computation is unavoidable without further assumptions. 
 
Another class of constraints arises when the right side of (\ref{eqn:plac}) vanishes identically rather than the left. The subsequent equation of motion depends only upon singly differentiated variables and is manifestly linear. Such constraints immediately imply that the time derivative of a linear combination of variables vanishes, or in other words a conservation law. As such, we refer to these constraints as \emph{Noether constraints}. Unlike null vector constraints, Noether constraints do not impede quantization, but nonetheless encode symmetries.

Now, we turn our attention toward Noether currents which arise due to the absence of some linear combination of degrees of freedom in the energy. Now, in the particular case of (\ref{eqn:no_omega_lagrangian}), we recognize that 
\begin{equation}
    \sum_{l \in L \subset\mathcal F} \alpha_l \frac{\partial E}{\partial q_l} = 0
\end{equation}
identically if and only if $\ket{L}$ with corresponding coefficients $\alpha_l \in \mathbb{R}$ satisfies 
\begin{equation}
    \mathbb P_{\mathcal C} |L \rangle = 0.
\end{equation}
If it is also the case that 
\begin{equation}\label{eqn:additional}
    \mathbb P_{\mathcal G} |L\rangle = 0, 
\end{equation}
then 
\begin{equation}
      \sum_{l \in L} \alpha_l \frac{\delta}{\delta q_l}\int\mathrm dt\,\Big(\langle q|M|\dot\phi\rangle + \sum_{\fg\in\fG}\mathscr{L}_\fg\Big) = 0
\end{equation}
vanishing off shell identically. 
If (\ref{eqn:additional}) fails to hold, then 
\begin{equation}\label{eqn:noetherloop}
    \sum_{l \in L \subset \mathcal F} \alpha_l\frac{\delta}{\delta q_l} \int\mathrm d t \,\mathscr L = \frac{1}{2} \left(
       \langle \dot \phi| A^\intercal \mathbb P_{\mathcal G} |L\rangle  - \sum_{\mathfrak g \in \fG} G_{\mathfrak g}^{-1}\,\langle L | \Gamma_{\mathfrak g}\,B^\intercal | \dot q\rangle
    \right)
\end{equation} 
Hence, for each loop set in a circuit with no capacitors ($\mathbb{P}_\mathcal{C} \ket{L} = 0$), there exists a Noether-charge constraint. Unlike the null-vector case, gyrator content plays a distinguishing role here: a Noether candidate becomes non-trivial precisely when it carries gyrator support. Specifically, if the loop in question is also free of gyrative edges ($\mathbb{P}_\mathcal G\ket{L}=0$), then the constraint is trivial as both sides of Eq. (\ref{eqn:plac}) vanish. A capacitor-free loop with $\mathbb{P}_\mathcal{G}\ket{L} \neq 0$ has a non-vanishing left-hand side and produces a non-trivial Noether current of the form (\ref{eqn:noetherloop}), which is recognizable as a sum of constitutive relations for gyrative edges. Explicitly, these non-trivial Noether constraints are generated by the $\ket{L}$ in the image of $B^\intercal$ such that:
\begin{subequations}
\begin{align}
        \mathbb{P}_{\mathcal C} \ket{L} &= 0,\\
        \mathbb{P}_{\mathcal G} \ket{L} &\neq 0
\end{align}
\end{subequations}
The conserved Noether current corresponds to $q_l \rightarrow q_l +c \alpha_l$ for every $l \in L$ which leaves the Lagrangian invariant. A dual argument provides Noether constraints on node fluxes. Considering instead every $U \subset \mathcal{V}$ such that $\ket U$ satisfies $\PP_{\mathcal I}\ket U = 0$ and $\PP_{\mathcal G}\ket U \neq 0$, the symmetry $\phi_v \to \phi_v + c \alpha_v$ ($v \in U$) leaves $\mathscr L$ invariant and removes a coordinate.

Up to this point, we have not at all considered the underlying topology of configuration space. 
In essence, the Lagrangian formalism we have introduced acts as a mnemonic for the integrating out of redundancies in configuration space, ultimately providing a space of reduced dimension that can be outfitted with a symplectic form. 

After isolating and solving all null vector constraints, a circuit Lagrangian may always be written as 
\begin{equation}\label{eqn:darboux}
    \mathscr L = \sum_{i} p_i \dot x_i - H(x,p)
\end{equation}
where $x_i$ and $p_i$ are pairwise independent linear combinations of $q$ and $\phi$ variables not necessarily indexed by any graph theoretic object.
Once the problem of identifying and solving constraints is resolved, the problem of quantizing nonreciprocal circuits is exactly the same as the problem of quantizing reciprocal circuits. 
The procedure of writing the Lagrangian in this way is equivalent to producing Darboux coordinates. 
The Legendre transform then identifies $x_i$ as the conjugate partner of $p_i$, with Hamiltonian function given by $H$.\footnote{Which is different from the total energy function only in that various constraints have been solved}
The physical properties of the system at hand do not follow from this information alone, however. 
In particular, since quantization is the goal, one must consider the fact that changes to the classical phase space can have measurable effects on the quantum Hamiltonian. 
It is tempting to simply write the commutation relations as 
\begin{equation}\label{eqn:canonical}
    [x_i, p_j] = \mathrm i \hbar \delta_{ij}, 
\end{equation}
but this can produce incorrect results. 
The simplest example illustrating this concept is the transmon qubit with 
\begin{equation}\label{eqn:tmon}
    H = \frac{1}{2 C}(q + q_g)^2 - E_J \cos(\phi)
\end{equation}
where $q_g$ is a constant and $q$ is the conjugate partner of $\phi$. 
If the given commutation relations are $[\phi, q] = \mathrm i \hbar$, then the spectrum of (\ref{eqn:tmon}) is gapless and independent of $q_g$; this choice implies that $\phi$ lives on $\mathbb R$ and hence Bloch's theorem applies.
On the other hand, if we demand that $\phi \sim \phi + c$ for some constant $c$, then we must have that $\cos(\phi + c) = \cos(\phi)$, so it follow that $c = 2 \pi n$. Because the operator $\phi$ is multiply valued on $S^1$, the commutation relations must instead be written as
\begin{equation}
    [e^{\mathrm i \phi / n}, q] = -\frac{\hbar}{n} e^{\mathrm i \phi/n}.
\end{equation}
Even the value of $n$ is physically important: for $n = 1$, the energy spacing between the first excited and ground states is $\sim \sqrt{E_J/C}$ while for $n = 2$ it is instead $\sim e^{-8\sqrt{E_J C}}$.
There is no mathematical reason by which to decide that some variable lives in $S^1$ as opposed to $\mathbb R$. 
Indeed, even in the case of the transmon, we distinguish between variables with a periodic identification and those without by experiment. 
In the matter of quantization, these statements play the same role with nonreciprocal circuits as with reciprocal ones. 
Explicitly, one must first write $\mathscr L$ in the form (\ref{eqn:darboux}), thus defining the pairs $x_i$ and $p_i$. 
Then, the quantum Hamiltonian function is given by $H(x,p)$ subject to commutation relations consistent with (\ref{eqn:canonical}) but written in terms of only singly valued functions on the chosen classical phase space. 

\subsection{Gyrator Replacement Rules}\label{sec:combo}

\begin{figure}
    \centering
    \begin{subfigure}[t]{0.36\linewidth}
        \centering
        \begin{circuitikz}[scale=2]
          \draw[thick] (0,-.5) to[twoport,t=$S_1$] (0,1.5)  to (1,1.5) to [half gyrator,i_=$e_1$,color=blue] (1,.6)to[half gyrator,i_=$e_2$,color=red] (1,-.4) to (1,-.5) to (0,-.5);
          \filldraw[black] (0.5,-.5) circle (1pt) node[anchor=north]{$v_3$};
          \filldraw[black] (0.5,1.5) circle (1pt) node[anchor=south]{$v_1$};
          \filldraw[black] (1,.5) circle (1pt) node[anchor=west]{$v_2$};
          \draw[thin, <-] (0.5,.5) node{$f_1$}  ++(-60:0.2) arc (-60:170:0.2);
          \draw[thin,->] (-.1,-.3) to[short,l_=$f_2$] (-.1,-.6) to (.2,-.6);
          \draw[thick] (1.7,-.5) to (2.7,-.5) to[twoport,t=$S_2$] (2.7,1.5) to (2.7,1.5) to (1.7,1.5) to[half gyrator, mirror,i_=$e_3$,color=blue] (1.7,.6) to[half gyrator, mirror,i_=$e_4$,color=red] (1.7,-.4) to (1.7,-.5);
          \filldraw[black] (2.2,-.5) circle (1pt) node[anchor=north]{$v_6$};
          \filldraw[black] (2.2,1.5) circle (1pt) node[anchor=south]{$v_4$};
          \filldraw[black] (1.7,.5) circle (1pt) node[anchor=east]{$v_5$};
          \draw[thin, <-] (2.2,.5) node{$f_3$}  ++(-60:0.2) arc (-60:170:0.2);
          \draw[thin,->] (2.8,1.3) to (2.8,1.6) to[short,l_=$f_4$] (2.5,1.6);
        \end{circuitikz}
        \caption{sub-circuits connected by gyrators in series}
        \label{fig:gyrator_combination_series}
    \end{subfigure}\hfill
    \begin{subfigure}[t]{0.62\linewidth}
        \centering
        \begin{circuitikz}[scale=2]
          \draw[thick] (-0.2,-.5) to[twoport,t=$S_1$] (-0.2,1.5)  to (0.8,1.5) to [half gyrator,i_=$e_1$,color=red] (0.8,-.4) to (0.8,-.5) to (-0.2,-.5); 
          \draw[thick] (0.8,1.5) to (1.8,1.5) to [half gyrator,i_=$e_2$,color=blue] (1.8,-.4) to (1.8,-.5) to (0.8,-.5);          
          \filldraw[black] (0.8,-.5) circle (1pt) node[anchor=north]{$v_2$};
          \filldraw[black] (0.8,1.5) circle (1pt) node[anchor=south]{$v_1$};
          \draw[thin, <-] (0.3,.5) node{$f_1$}  ++(-60:0.2) arc (-60:170:0.2);
          \draw[thin, <-] (1.3,.5) node{$f_2$}  ++(-60:0.2) arc (-60:170:0.2);
          \draw[thin,->] (-.3,-.3) to[short,l_=$f_3$] (-.3,-.6) to (-.1,-.6);
          \draw[thick] (2.5, 1.5) to [half gyrator, mirror,i=$e_3$,color=blue] (2.5,-.4) to (2.5,-.5);
          \draw[thick] (3.5, 1.5) to [half gyrator, mirror,i=$e_4$,color=red] (3.5,-.4) to (3.5,-.5);
          \draw[thick] (2.5,1.5) to (3.5,1.5);
          \draw[thick] (2.5,-.5) to (3.5,-.5);
          \draw[thick] (3.5,1.5) to (4.5,1.5) to [twoport,t=$S_2$] (4.5,-.5) to (3.5,-.5);          
          \filldraw[black] (3.5,-.5) circle (1pt) node[anchor=north]{$v_4$};
          \filldraw[black] (3.5,1.5) circle (1pt) node[anchor=south]{$v_3$};
          \draw[thin, <-] (3,.5) node{$f_4$}  ++(-60:0.2) arc (-60:170:0.2);
          \draw[thin, <-] (4,.5) node{$f_5$}  ++(-60:0.2) arc (-60:170:0.2);
          \draw[thin,->] (2.4,-.3) to[short,l_=$f_6$] (2.4,-.6) to (2.7,-.6);
        \end{circuitikz}
        \caption{sub-circuits connected by gyrators in parallel}
        \label{fig:gyrator_combination_parallel}
    \end{subfigure}
    \caption{The labeled squares denote two-port devices with some finite number of internal elements. Unlike later sections, we allow sub-circuits $S_1$ and $S_2$ here to contain reciprocal elements and gyrative edges. In both cases presented, the two gyrators, with gyration ratios $G_1$ and $G_2$, respectively, in each circuit may be represented as a single gyrator with gyration ratio $G^*$. In (a) $G^* = (\frac{1}{G_1} + \frac{1}{G_2})^{-1}$, while in (b) $G^* = G_1 + G_2$.}\label{fig:gyrator_combination}
\end{figure}
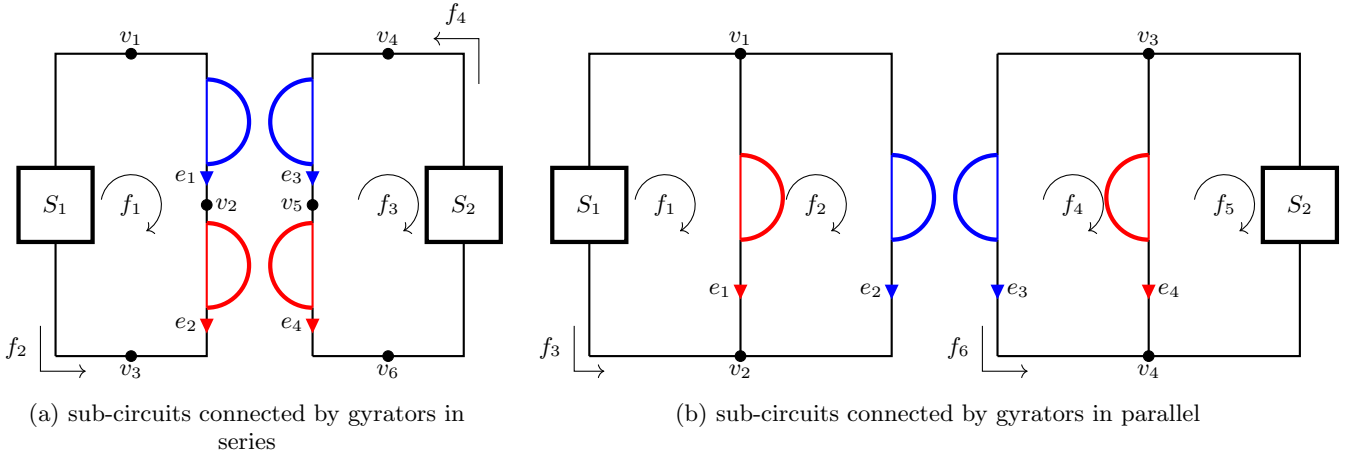

There are notable configurations of two gyrators that may be replaced by a simpler element. We present three such rules in this section: the  \textit{series} and \textit{parallel deletion rules}, which each combine two  gyrators into a single gyrator with effective gyration ratio $G^*$, and the \textit{cascade-to-transformer rule}, which replaces two gyrators sharing an isolated RCC with a single ideal transformer. All three are point transformations on the Lagrangian, and therefore canonical transformations after the Legendre transform. All of these work by identifying a topologically isolated vertex, face, or RCC, whose dynamical variables appear only in the gyrator terms and can be eliminated by their EOM. 

First, we derive the series deletion rule. Consider the circuit in Figure \ref{fig:gyrator_combination_series} and assume it is nonsingular and well-posed. We draw attention in particular to gyrators $\mathfrak{g}_1 = (e_1, e_3)$ and $\mathfrak{g}_2 = (e_2, e_4) \in \mathfrak{G}$, with gyration ratios $G_1$ and $G_2$ respectively. In particular, vertex $v_2$ is incident only on $e_1$ and $e_2$ and $v_5$ on $e_3$ and $e_4$. Therefore, $\phi_{v_2}$ and $\phi_{v_5}$ appear only in the Lagrangian terms for the two gyrators:
\begin{equation}
\mathscr L_{\mathfrak g_1, \mathfrak g_2} 
  = \tfrac{G_1}{2}(\phi_{v_2} - \phi_{v_1})(\dot\phi_{v_5} - \dot\phi_{v_4}) 
   + \tfrac{G_2}{2}(\phi_{v_3} - \phi_{v_2})(\dot\phi_{v_6} - \dot\phi_{v_5}) - \tfrac{1}{2}\!\left(\tfrac{1}{G_1} + \tfrac{1}{G_2}\right)\!
        (q_{f_1} - q_{f_2})(\dot q_{f_4} - \dot q_{f_3}).
\end{equation}
We identify the two Noether constraints obtained from the functional derivatives with respect to $\phi_{v_2}$ and $\phi_{v_5}$, which appear only in the gyrator terms, to get:
\begin{subequations}
\begin{align}
    \frac{\delta S}{ \delta \phi_2} &= \frac{G_1}{2}(\dot\phi_{v_5} - \dot\phi_{v_4} ) -  \frac{G_2}{2}(\dot\phi_{v_6} - \dot\phi_{v_5} ) = 0   \\
    \frac{\delta S}{ \delta \phi_5} &= \frac{G_1}{2}(\dot\phi_{v_2} - \dot\phi_{v_1} ) -  \frac{G_2}{2}(\dot\phi_{v_3} - \dot\phi_{v_2} ) = 0 
\end{align}
\end{subequations}
Integrating both equations and gauge-fixing the integration constants to zero solves for both vertex fluxes:
\begin{subequations}
\begin{align}
      \phi_{v_2} &= \frac{G_1 \phi_{v_1} + G_2 \phi_{v_3}}{G_1 + G_2}\\ 
  \phi_{v_5} &= \frac{G_1 \phi_{v_4} + G_2 \phi_{v_6}}{G_1 + G_2}.
\end{align}
\end{subequations}
Substituting both back into gyrator terms collapses them into a single effective gyrator contribution:
\begin{equation}
\mathscr L_{\mathfrak g^*} 
  = \tfrac{G^*}{2}(\phi_{v_3} - \phi_{v_1})(\dot\phi_{v_6} - \dot\phi_{v_4})
  - \tfrac{1}{2 G^*}(q_{f_1} - q_{f_2})(\dot q_{f_4} - \dot q_{f_3})
\end{equation}
where $G^* = \left(\frac{1}{G_1} + \frac{1}{G_2}\right)^{\!-1}$. Therefore, a circuit with two gyrators where one half of each gyrator is in series with one half of the other may be expressed as a single gyrator.

Next, we derive the parallel deletion rule. Consider the circuit in Figure \ref{fig:gyrator_combination_parallel} and assume it is nonsingular and well-posed. We draw attention in particular to gyrators $\mathfrak{g}_1 = (e_1, e_4)$ and $\mathfrak{g}_2 = (e_2, e_3) \in \mathfrak{G}$, with gyration ratios $G_1$ and $G_2$ respectively. We write the Lagrangian contribution of the two gyrators:
\begin{equation}
\mathscr L_{\mathfrak g_1, \mathfrak g_2} 
  = \tfrac{1}{2}(G_1 + G_2)(\phi_{v_2} - \phi_{v_1})(\dot\phi_{v_4} - \dot\phi_{v_3})
   - \tfrac{1}{2G_1}(q_{f_1} - q_{f_2})(\dot q_{f_4} - \dot q_{f_5})
   - \tfrac{1}{2G_2}(q_{f_2} - q_{f_3})(\dot q_{f_6} - \dot q_{f_4}).
\end{equation}
We recognize the two Lagrange multipliers and take functional derivatives with respect to $q_{f_2}$ and $q_{f_4}$ to get:
\begin{subequations}
\begin{align}
    \frac{\delta S}{\delta q_{f_2}} &= \tfrac{1}{2 G_1}(\dot q_{f_4} - \dot q_{f_5}) - \tfrac{1}{2 G_2}(\dot q_{f_6} - \dot q_{f_4}) = 0 \\
    \frac{\delta S}{\delta q_{f_4}} &= \tfrac{1}{2 G_1}(\dot q_{f_1} - \dot q_{f_2}) - \tfrac{1}{2 G_2}(\dot q_{f_2} - \dot q_{f_3}) = 0
\end{align}
\end{subequations}
Integrating both equations and gauge-fixing the integration constants to zero solves for both charges:
\begin{subequations}
\begin{align}
    q_{f_2} &= \frac{G_2 q_{f_1} + G_1 q_{f_3}}{G_1 + G_2} \\
    q_{f_4} &= \frac{G_2 q_{f_5} + G_1 q_{f_6}}{G_1 + G_2}.
\end{align}
\end{subequations}
Again, substituting both back into the gyrator terms yields a single effective gyrator contribution:
\begin{equation}
\mathscr L_{\mathfrak g^*} 
  = \tfrac{G^*}{2}(\phi_{v_2} - \phi_{v_1})(\dot\phi_{v_4} - \dot\phi_{v_3})
   - \tfrac{1}{2 G^*}(q_{f_1} - q_{f_3})(\dot q_{f_6} - \dot q_{f_5})
\end{equation}
where $G^* = G_1 + G_2$.

One may then consider a circuit where two gyrators have one half of their edge pairs in series with each other and the other half in parallel and find a pathological circuit configuration. As stated in Section \ref{sec:assumptions}, we maintain that circuits whose elements prevent dynamics are singular and, as shown below, such a configuration expands the class of singular circuits. Looking at Figure \ref{fig:singulargyrator}, the gyrators' constitutive equations give
\begin{subequations}\label{eqn:badgyr}
    \begin{align}
    G_1(\dot{\phi}_1-\dot{\phi}_2)=\dot{q}_5 -\dot{q}_3  \\
    -G_1(\dot{\phi}_5-\dot{\phi}_4)=\dot{q}_1 -\dot{q}_2 \\
    G_2(\dot{\phi}_2-\dot{\phi}_3)=\dot{q}_3 -\dot{q}_4 \\
    -G_2(\dot{\phi}_5-\dot{\phi}_4)=\dot{q}_1 -\dot{q}_2 
\end{align}
\end{subequations}

The second and fourth equations in Equations \ref{eqn:badgyr} are satisfied when $G_1=G_2$ or $\dot{\phi}_5-\dot{\phi}_4=\dot{q}_1 -\dot{q}_2=0$. In the former case, this pair then behaves as a single gyrator between the series and parallel port with gyration ratio $G_1$. In the latter case, however, there is no voltage across the parallel pair and no current through the series pair, making one a short and the other an open. No dynamics may propagate across the gyrators. If for example $S_1$ and $S_2$ were both composed of only capacitors (or inductors) then the lack of coupling would mean the whole circuit has no dynamics. Naturally then, these circuits have an intrinsic obstruction to a Hamiltonian description that does not depend on the method of producing a Lagrangian. 

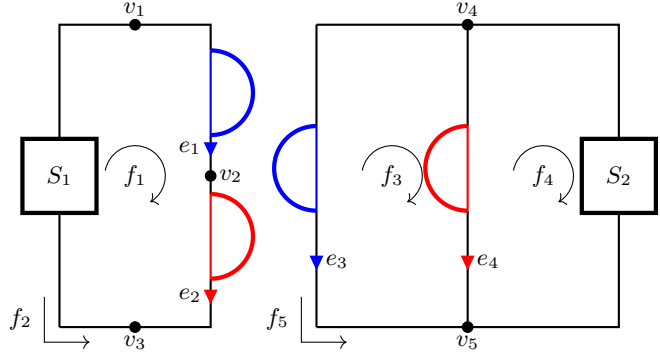
\begin{figure}
    \centering
    \begin{circuitikz}[scale=2]
        \draw[thick] (0,-.5) to[twoport,t=$S_1$] (0,1.5)  to (1,1.5) to [half gyrator,i_=$e_1$,color=blue] (1,.6)to[half gyrator,i_=$e_2$,color=red] (1,-.4) to (1,-.5) to (0,-.5);
        \filldraw[black] (0.5,-.5) circle (1pt) node[anchor=north]{$v_3$};
        \filldraw[black] (0.5,1.5) circle (1pt) node[anchor=south]{$v_1$};
        \filldraw[black] (1,.5) circle (1pt) node[anchor=west]{$v_2$};

        \draw[thin, <-] (0.5,.5) node{$f_1$}  ++(-60:0.2) arc (-60:170:0.2);
        \draw[thin,->] (-.1,-.3) to[short,l_=$f_2$] (-.1,-.6) to (.2,-.6);

        \draw[thick] (1.7, 1.5) to [half gyrator, mirror,i=$e_3$,color=blue] (1.7,-.4) to (1.7,-.5);
        \draw[thick] (2.7, 1.5) to [half gyrator, mirror,i=$e_4$,color=red] (2.7,-.4) to (2.7,-.5);
        \draw[thick] (1.7,1.5) to (2.7,1.5);
        \draw[thick] (1.7,-.5) to (2.7,-.5);

        \draw[thick] (2.7,1.5) to (3.7,1.5) to [twoport,t=$S_2$] (3.7,-.5) to (2.7,-.5);          
        \filldraw[black] (2.7,-.5) circle (1pt) node[anchor=north]{$v_5$};
        \filldraw[black] (2.7,1.5) circle (1pt) node[anchor=south]{$v_4$};
        \draw[thin, <-] (2.2,.5) node{$f_3$}  ++(-60:0.2) arc (-60:170:0.2);
        \draw[thin, <-] (3.2,.5) node{$f_4$}  ++(-60:0.2) arc (-60:170:0.2);
        \draw[thin,->] (1.6,-.3) to[short,l_=$f_5$] (1.6,-.6) to (1.9,-.6);
    \end{circuitikz}
    \caption{A singular circuit with gyrators. The constitutive equations in Equations \ref{eqn:badgyr} refer to the gyration ratio of the blue gyrator as $G_1$ and that of the red gyrator as $G_2$.}
    \label{fig:singulargyrator}
\end{figure}

Finally, we consider gyrators in cascade which reduce to a transformer, a reciprocal element \cite{tellegen}. For two gyrators with ratios $G_1$ and $G_2$ joined through an RCC containing no other circuit elements, the resulting transformer has turns ratio $n = G_1/G_2$. To demonstrate this, consider the left circuit in Figure \ref{fig:trans}. Assume it is a nonsingular, well-posed LCG circuit and that its two gyrators $\mathfrak{g}_1 = (e_1, e_1'), \mathfrak{g}_2 = (e_2, e_2') \in \mathfrak{G}$ with gyration ratios $G_{\mathfrak{g}_1} = G_1$ and $G_{\mathfrak{g}_2} = G_2$. Note the inner gyrative edges $e_1', e_2$ are in parallel (and in series as well), joined at the same two vertices. They are incident on no other edges of $C$. Then by an algebraic elimination of the constrained junction variables, we may replace $\mathfrak{g}_1$ and $\mathfrak{g}_2$ by a single ideal transformer relating $e_1$ and $e_2'$ with turns ratio $n = G_1/G_2$, and the inner gyrative edges $e_1', e_2$ are removed. The constitutive relations of the two gyrators are
\begin{subequations}
\begin{align}
    G_1 \dot\Phi_{e_1} &= \dot Q_{e_1'}\\
    G_2 \dot\Phi_{e_2} &= \dot Q_{e_2'},
\end{align}
\end{subequations}
together with the dual relations $G_1 \dot\Phi_{e_1'} = -\dot Q_{e_1}$ and $G_2 \dot\Phi_{e_2'} = -\dot Q_{e_2}$ that follow from the antisymmetry of the gyrator term in the Lagrangian. The series condition enforces $\dot Q_{e_1'} = -\dot Q_{e_2}$ as they share the same faces while $\dot\Phi_{e_1'} = \dot\Phi_{e_2}$ as the gyrative edges are incident on the same vertices. Substituting,
\begin{equation}
    G_1 \dot\Phi_{e_1} = \dot Q_{e_1'} = -\dot Q_{e_2} = G_2 \dot\Phi_{e_2'},
\end{equation}
so $\dot\Phi_{e_2'} = (G_1/G_2)\dot\Phi_{e_1}$. Similarly,
\begin{equation}
    \dot Q_{e_1} = -G_1 \dot\Phi_{e_1'} = -G_1 \dot\Phi_{e_2} = -(G_1/G_2)\dot Q_{e_2'},
\end{equation}
so $\dot Q_{e_1} = -(G_1/G_2)\dot Q_{e_2'}$. These are the constitutive relations of an ideal transformer relating $e_1$ and $e_2'$ with turns ratio $n = G_1/G_2$.

The intermediate variables $Q_{e_1'} = -Q_{e_2}$ and $\Phi_{e_1'} = \Phi_{e_2}$ have been eliminated by these substitutions and may be removed from the dynamical description. The Lagrangian of $C$ thereby reduces to the Lagrangian of a circuit $C'$ in which $\mathfrak{g}_1$ and $\mathfrak{g}_2$ are replaced by the transformer just described, $e_1', e_2$ are removed, and all other elements are unchanged.  As transformers appear in our LCG framework as composite gyrator-gyrator objects, we have a mechanism for converting two non-reciprocal elements into a single reciprocal one (a fact we exploit in Section \ref{sec:fundrecip}).

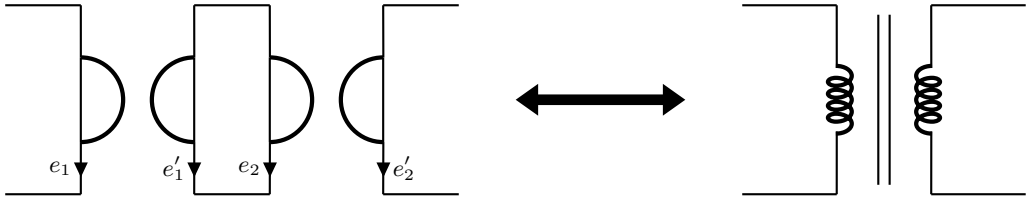
\begin{figure}[h]
    \centering
    \begin{circuitikz}[scale=2.5]
        \draw[thick] (0, 0) to (0.4, 0);
        \draw[thick] (0, 1) to (0.4, 1);
        \draw[thick] (0.4, 0) to[half gyrator,mirror, i<=$e_1$] (0.4, 1);
        \draw[thick] (1.0, 0) to[half gyrator,i<=$e_1'$] (1.0, 1);
        \draw[thick] (1.0, 0) to (1.4, 0);
        \draw[thick] (1.0, 1) to (1.4, 1);
        \draw[thick] (1.4,0) to [half gyrator, mirror,i<=$e_2$](1.4,1);
        \draw[thick] (2,1) to [half gyrator,mirror,i>=$e_2'$](2,0);
        \draw[thick] (2, 0) to (2.4, 0);
        \draw[thick] (2, 1) to (2.4, 1);

        \draw[{Triangle[length=3mm, width=4mm]}-{Triangle[length=3mm, width=4mm]}, line width=4pt]
            (2.7, 0.5) -- (3.6, 0.5);
 
        \begin{scope}[shift={(3.9, 0)}]
            \draw[thick] (0, 0) to (0.5, 0);
            \draw[thick] (0, 1) to (0.5, 1);
            \draw[thick] (0.5, 0) to[L, mirror] (0.5, 1);
            \draw[thick] (1.0, 0) to[L]      (1.0, 1);
            \draw[thick] (0.72, 0.05) -- (0.72, 0.95);
            \draw[thick] (0.78, 0.05) -- (0.78, 0.95);
            \draw[thick] (1.0, 0) to (1.5, 0);
            \draw[thick] (1.0, 1) to (1.5, 1);
        \end{scope}
    \end{circuitikz}
    \caption{Two gyrators in cascade through a junction containing no reciprocal elements (left) are dynamically equivalent to a single ideal transformer with turns ratio $G_1/G_2$ (right). 
    }
    \label{fig:trans}
\end{figure}

The construction of an ideal transformer in our formalism is a minimal example of the networks that may be realized with gyrators. Gyrators in combination with reciprocal network elements can synthesize arbitrary linear passive non-reciprocal networks \cite{carlin}. This result rigorously established that circulators and other common non-reciprocal elements can be realized with no other source of non-reciprocity than that which can be modeled by a gyrator \cite{carlin}. Specifically, such a network is realizable if and only if the Hermitian part of its impedance matrix is positive-definite or positive semi-definite \cite{carlin}. Combined with our Hamiltonian formalism, we may predict the behavior of networks of linear non-reciprocal elements with the addition of non-linear conservative elements.

\section{Duality and  Circuit Equivalence Transformations}\label{sec:transformations} 
\subsection{Circuit Duality}
    Expanding on the construction of \cite{fluxcharge}, we define a notion of duality for our LCG circuits: an involutive mapping on circuits that acts unitarily on the Hamiltonian. Consider the following transformation which will henceforth be known as the LCG duality transform:
    \begin{equation}
        \begin{aligned}
            \phi_v &\rightarrow \phi_v^* = q'_l \\
            q_l &\rightarrow q_l^* = \phi'_v \\
            A &\rightarrow A^* = B^\intercal \\
            B &\rightarrow B^* = A^\intercal \\
            \mathcal{C} &\rightarrow \mathcal{C^*} \simeq \mathcal{I} \\
            \mathcal{I} &\rightarrow \mathcal{I^*} \simeq \mathcal{C} \\
            \mathcal{G} &\rightarrow \mathcal{G^*} \simeq \mathcal{G} \\
            \fG &\rightarrow \fG^* = \fG\\ G_\mathfrak{g} &\rightarrow G_{\mathfrak g}^* = -G_{\mathfrak g}^{-1} \\
            \mathcal{V} &\rightarrow \mathcal{V^*} \simeq \mathcal{F} \\
            \mathcal{F} &\rightarrow \mathcal{F^*} \simeq \mathcal{V}
        \end{aligned}
    \end{equation}

The $\Gamma_{\mathfrak g}$ depends only on the edge pair $\mathfrak{g}$, so it is preserved under duality, so the gyrator's contribution to the Lagrangian transforms entirely through the gyration-ratio rule $G_{\mathfrak g} \to -G_{\mathfrak g}^{-1}$. The substitutions $\phi_v \to q^*_l$ and $q_l \to \phi^*_v$ have significance to the units in SI: node fluxes carry units of magnetic flux while loop charges carry units of electric charge, so the naive substitution would render energy terms dimensionally inconsistent. A capacitive energy $Q_e^2/(2C)$ mapped term-by-term to $(\Phi^*_e)^2/(2C)$ no longer has units of energy. 
The unit convention we establish in Section \ref{sec:assumptions} which demands a reference conductance $G_0 = 1$ ensures physically sensible dimensions. In this convention flux and charge have the same units, as do their linear coupling coefficients inductance and capacitance.
Hence, the per-element rule for reciprocal components is then simply $C^*_e = L_e$ and $L^*_e = C_e$, preserving the numerical value of each scalar parameter across the dual map. This numerical preservation is a property of the global duality transform in the unit convention $G_0 = 1$. The partial dual of Section \ref{sec:rules}, which exchanges flux and charge across one gyrator of ratio $G_\fg$, rescales the element parameters by $G_\fg^{\pm 2}$, Eq. \ref{eqn:elementrules}. The gyration ratios transform as $G_{\mathfrak g} \to -G_{\mathfrak g}^{-1}$, ensuring consistency of units. Exchanging flux and charge turns $\tfrac{G_\mathfrak{g}}{2}\Phi_{e_1}\dot\Phi_{e_2} - \tfrac{1}{2G_\mathfrak{g}}Q_{e_1}\dot Q_{e_2}$ into $\tfrac{G_\mathfrak{g}}{2}Q_{e_1}\dot Q_{e_2} - \tfrac{1}{2G_\mathfrak{g}}\Phi_{e_1}\dot\Phi_{e_2}$, which is again a gyrator term, now with ratio $-G_{\mathfrak g}^{-1}$. The inversion follows from the fact that a gyration ratio has units of charge per flux while the sign change reflects the antisymmetry of the gyrator constitutive relations. The same per-element rule applies regardless of how many gyrators the circuit contains, with each gyration ratio transforming independently. We now consider the transformation on each term of the Lagrangian. From $A \to B^\intercal$, $B \to A^\intercal$, $\mathcal C \leftrightarrow\mathcal I$, and $\mathcal G \to \mathcal G$, the dual connection matrix is:
\begin{equation}\label{eqn:Mstar}
M = \tfrac{1}{2} B(\mathbb P_{\mathcal C} - \mathbb P_{\mathcal I}) A
\;\longrightarrow\;
M^* = \tfrac{1}{2}A^\intercal(\mathbb P_{\mathcal I} - \mathbb P_{\mathcal C}) B^\intercal
= -\tfrac{1}{2} A^\intercal(\mathbb P_{\mathcal C} - \mathbb P_{\mathcal I})B^\intercal
= -M^\intercal,
\end{equation}
The following matrices are swapped under duality:
\begin{subequations}\label{eqn:Gammastar}
\begin{align}
    A^\intercal \Gamma_{\mathfrak g} A &\longrightarrow\; B\Gamma_{\mathfrak g} B^\intercal\\
    B\Gamma_{\mathfrak g} B^\intercal &\longrightarrow\; A^\intercal\Gamma_{\mathfrak g} A
\end{align}
\end{subequations}
Considering the above and the transformation of the gyration ratios $G_{\mathfrak g} \to -G_{\mathfrak g}^{-1}$, the per-gyrator contribution is similarly invariant:
\begin{equation}\label{eqn:Lgstar}
\begin{aligned} 
    \mathscr L_{\mathfrak g} &= \frac{G_{\mathfrak g}}{4}\langle\phi|A^\intercal \Gamma_{\mathfrak g} A|\dot\phi\rangle \;-\;\frac{1}{4 G_{\mathfrak g}}\langle q|B \Gamma_{\mathfrak g} B^\intercal|\dot q\rangle\\
    \longrightarrow \mathscr L_{\mathfrak g}^* &= \frac{G_{\mathfrak g}^{*}}{4}\langle \phi^*|(A^*)^\intercal \Gamma_{\mathfrak g} A^*|\dot \phi^*\rangle \;-\;  \frac{1}{4G_{\mathfrak g}^{*}}\langle q^*|B^* \Gamma_{\mathfrak g} (B^*)^\intercal|\dot q^*\rangle \\
    &= \frac{ (-G_{\mathfrak g}^{-1})}{4}\langle q |B \Gamma_{\mathfrak g} B^\intercal|\dot q \rangle \;-\;\frac{1}{4(-G_{\mathfrak g}^{-1})}\langle \phi|A^\intercal \Gamma_{\mathfrak g} A|\dot \phi\rangle \\
    &= -\frac{ 1}{4G_{\mathfrak g}}\langle q |B \Gamma_{\mathfrak g} B^\intercal|\dot q \rangle \;+\; \frac{G_{\mathfrak g}}{4}\langle\phi|A^\intercal \Gamma_{\mathfrak g} A|\dot\phi\rangle\\
    &=  \mathscr L_{\mathfrak g}\\
\end{aligned}
\end{equation}
Finally we consider the energies of the reciprocal elements. These depend only on the edge charge or edge flux and as such a capacitive element $E^C_e(Q_e)$ with $Q_e = \bra{q}B\ket{e}$ undergoes the following transformation after exchanging its element type and $Q_e \rightarrow \Phi_e^*$:
\begin{equation}\label{eqn:ECstar}
E_e^C \;\longrightarrow\; (E_e^C)^* = E_e^I\big|_{e \in \mathcal I^*}
\end{equation}
Similarly, an inductive element has its energy term transformed as:
\begin{equation}\label{eqn:EIstar}
E_e^I \;\longrightarrow\; (E_e^I)^* = E_e^C\big|_{e \in \mathcal C^*}
\end{equation}
Substituting \eqref{eqn:Mstar}, \eqref{eqn:Lgstar}, \eqref{eqn:ECstar}, and \eqref{eqn:EIstar}
into $\mathscr L$, the Lagrangian then transforms as:
\begin{equation}\label{eqn:Lstar}
\mathscr L
= \langle q|M|\dot\phi\rangle
   + \sum_{\mathfrak g \in \fG} \mathscr L_{\mathfrak g}
   - \sum_{e \in \mathcal C} E_e^C
   - \sum_{e \in \mathcal I} E_e^I
\longrightarrow\
\mathscr L^*
= -\langle\phi|M^\intercal|\dot q\rangle
   + \sum_{\mathfrak g \in \fG} \mathscr L_{\mathfrak g}
   - \sum_{e \in \mathcal I} E_e^I
   - \sum_{e \in \mathcal C} E_e^C.
\end{equation}

After integrating the first term of $\LL^*$ by parts (boundary contributions vanish at the temporal endpoints), one may note $-\bra{\phi}M^\intercal\ket{\dot q} = \bra{q}M\ket{\dot\phi} - \tfrac{\mathrm d}{\mathrm dt}\bra{q}M\ket{\phi}$, so $\LL^* = \LL$ up to a total time derivative. Therefore the action, the equations of motion and the Hamiltonian spectrum are invariant under the LCG duality transform. As in previous work, the duality mapping relies on the really-connected components of the underlying circuit graph being planar. Given that the number of faces is exactly one more (counting the outer face) than the number of linearly independent loops in a really-connected component, $A^{(i)*} \to B^{(i)\intercal}$ uniquely provides valid boundary maps that satisfy Kirchhoff's laws. As this transformation preserves the kernel of the pre-symplectic matrix, it preserves the Hamiltonian. 

\subsubsection{Component-Wise Duality Transformation}\label{sec:cwd}
For a circuit with many really-connected components, we may apply the duality transformation componentwise, that is, only applying the dual operation to a particular really-connected component while preserving the spectrum. As the boundary maps decompose as $A = \bigoplus_i A^{(i)}$ and $B = \bigoplus_i B^{(i)}$ (Equation \ref{eqn:bmapdecomp}), the connection matrix $M$ is block-diagonal over really-connected components, and the substitution above may be applied to any chosen subset of the components while the remainder are left untouched. We emphasize that this transformation is a relabeling of coordinates that leaves the Lagrangian and EL equations identical. Terms internal to the untransformed components remain the same while gyrators inside the transformed components entirely have their ratio changed to $-G_\fg^{-1}$. However, any interaction between transformed and untransformed really-connected components is mediated by the transformation of a gyrator between a transformed really-connected component and an untransformed really-connected component into an ideal transformer. To clarify the previous sentence we provide an example in Figure \ref{fig:cwd_example}.

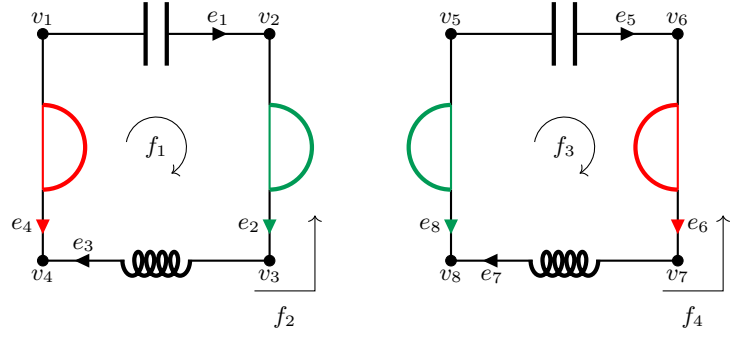
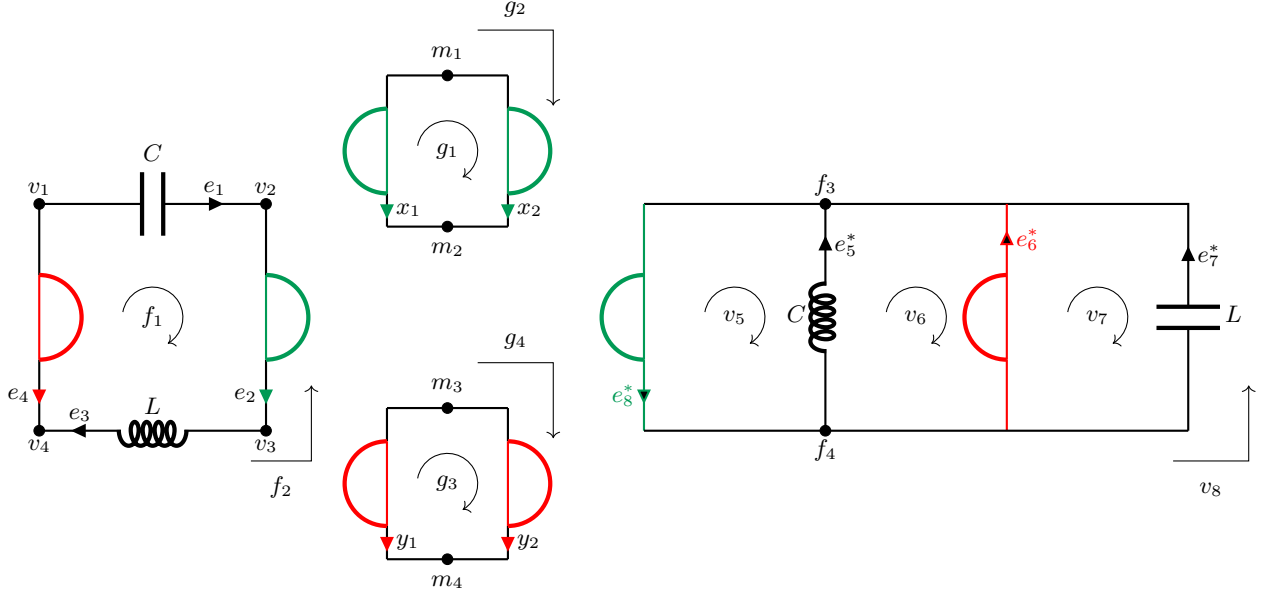
\begin{figure}[h]
    \centering
    \begin{subfigure}{\linewidth}
        \centering
        \begin{circuitikz}[scale=2.0]
            \path (0,-0.65) (0,1.95);
            \draw[thick] (0,1.5) to[half gyrator, color=red, i_=$e_4$] (0,0);
            \draw[thick] (0,1.5) to[capacitor,i=$e_1$] (1.5,1.5) to[half gyrator,color=ForestGreen, i_=$e_2$] (1.5,0)
            to [inductor,i_=$e_3$](0,0);
            \filldraw[black] (0,1.5) circle (1pt) node[anchor=south]{$v_{1}$};
            \filldraw[black] (1.5,1.5) circle (1pt) node[anchor=south]{$v_{2}$};
            \filldraw[black] (1.5,0) circle (1pt) node[anchor=north]{$v_{3}$};
            \filldraw[black] (0,0) circle (1pt) node[anchor=north]{$v_{4}$};
            \draw[thin, <-] (.75,.75) node{$f_1$} ++(-60:0.2) arc (-60:170:0.2);
            \draw[thin,->] (1.4,-0.2) to[short,l_=$f_2$] (1.8,-0.2) to (1.8,0.3);
            \begin{scope}[shift={(2.7,0)}]
                \draw[thick] (0,1.5) to[half gyrator,color=ForestGreen,mirror, i_=$e_8$] (0,0);
                \draw[thick] (0,1.5) to[capacitor,i=$e_5$] (1.5,1.5);
                \draw[thick] (1.5,1.5) to [half gyrator, mirror, color=red,i=$e_6$] (1.5,0);
                \draw[thick] (1.5,0) to [inductor,i=$e_7$](0,0);
                \filldraw[black] (0,1.5) circle (1pt) node[anchor=south]{$v_{5}$};
                \filldraw[black] (1.5,1.5) circle (1pt) node[anchor=south]{$v_{6}$};
                \filldraw[black] (1.5,0) circle (1pt) node[anchor=north]{$v_{7}$};
                \filldraw[black] (0,0) circle (1pt) node[anchor=north]{$v_{8}$};
                \draw[thin, <-] (.75,.75) node{$f_3$} ++(-60:0.2) arc (-60:170:0.2);
                \draw[thin,->] (1.4,-0.2) to[short,l_=$f_4$] (1.8,-0.2) to (1.8,0.3);
            \end{scope}
        \end{circuitikz}\caption{  A circuit with two gyrators.}\label{fig:cwd_a}
        
    \end{subfigure}

    \begin{subfigure}{\linewidth}
        \centering
        \begin{circuitikz}[scale=2.0]
            \path (0,-1.45) (0,2.95);
            \draw[thick] (0,1.5) to[half gyrator, color=red, i_=$e_4$] (0,0);
            \draw[thick] (0,1.5) to[capacitor,i=$e_1$,l=$C$] (1.5,1.5) to[half gyrator,color=ForestGreen, i_=$e_2$] (1.5,0)
            to [inductor,i_=$e_3$,l_=$L$](0,0);
            \filldraw[black] (0,1.5) circle (1pt) node[anchor=south]{$v_{1}$};
            \filldraw[black] (1.5,1.5) circle (1pt) node[anchor=south]{$v_{2}$};
            \filldraw[black] (1.5,0) circle (1pt) node[anchor=north]{$v_{3}$};
            \filldraw[black] (0,0) circle (1pt) node[anchor=north]{$v_{4}$};
            \draw[thin, <-] (.75,.75) node{$f_1$} ++(-60:0.2) arc (-60:170:0.2);
            \draw[thin,->] (1.4,-0.2) to[short,l_=$f_2$] (1.8,-0.2) to (1.8,0.3);

            \draw[thick] (2.3,2.35) to (3.1,2.35);
            \draw[thick] (2.3,1.35) to (3.1,1.35);
            \draw[thick] (2.3,2.35) to [half gyrator, mirror, color=ForestGreen, i=$x_1$] (2.3,1.35);
            \draw[thick] (3.1,2.35) to [half gyrator, color=ForestGreen, i=$x_2$] (3.1,1.35);
           
            \filldraw[black] (2.7,2.35) circle (1pt);
            \filldraw[black] (2.7,1.35) circle (1pt);
            \node[anchor=south] at (2.7,2.4) {$m_{1}$};
            \node[anchor=north] at (2.7,1.3) {$m_{2}$};
            \draw[thin, <-] (2.7,1.85) node{$g_1$} ++(-60:0.2) arc (-60:170:0.2);
            \draw[thin,->] (2.9,2.65) to[short,l=$g_2$] (3.4,2.65) to (3.4,2.15);

            \draw[thick] (2.3,.15) to (3.1,.15);
            \draw[thick] (2.3,-.85) to (3.1,-.85);
            \draw[thick] (2.3,.15) to [half gyrator, mirror, color=red, i=$y_1$] (2.3,-.85);
            \draw[thick] (3.1,.15) to [half gyrator, color=red, i=$y_2$] (3.1,-.85);
           
            \filldraw[black] (2.7,.15) circle (1pt);
            \filldraw[black] (2.7,-.85) circle (1pt);
            \node[anchor=south] at (2.7,.2) {$m_{3}$};
            \node[anchor=north] at (2.7,-.9) {$m_{4}$};
            \draw[thin, <-] (2.7,-.35) node{$g_3$} ++(-60:0.2) arc (-60:170:0.2);
            \draw[thin,->] (2.9,.45) to[short,l=$g_4$] (3.4,.45) to (3.4,-.05);

            \begin{scope}[shift={(7,0)}]
            \draw[thick] (0.6, 0) to[capacitor,mirror, i_=$e_7^*$,l_=$L$] (0.6,1.5)
            to (-3,1.5);
            \draw[thick, ForestGreen] (-3,1.5) to [half gyrator, mirror, i_=$e_8^*$] (-3,0);
            \draw[thick] (-1.8,0) to[inductor, i_=$e_5^*$,l=$C$] (-1.8,1.5);
            \draw[thick, red] (-0.6,0) to[half gyrator, i_=$e_6^*$] (-0.6,1.5);
            \draw[thick] (-3,0) to (0.6,0);
            \filldraw[black] (-1.8,1.5) circle (1pt) node[anchor=south]{$f_3$};
            \filldraw[black] (-1.8,0) circle (1pt) node[anchor=north]{$f_4$};
            \draw[thin, <-] (-2.4,.75) node{$v_5$} ++(-60:0.2) arc (-60:170:0.2);
            \draw[thin, <-] (-1.2,.75) node{$v_6$} ++(-60:0.2) arc (-60:170:0.2);
            \draw[thin, <-] (0,.75) node{$v_7$} ++(-60:0.2) arc (-60:170:0.2);
            \draw[thin,->] (0.5,-0.2) to[short,l_=$v_8$] (1.0,-0.2) to (1.0,0.3);
            \end{scope}
        \end{circuitikz}
        \caption{ The isospectral circuit realized after taking the component-wise dual of the circuit in (a).}\label{fig:cwd_b}
    \end{subfigure}
   \caption{Two isospectral circuits related by the application of a component-wise duality transformation.} \label{fig:cwd_example}
\end{figure}

Consider first the circuit presented in Figure \ref{fig:cwd_a}.  The green gyrative edges form the gyrator $\mathfrak g_1 = (e_2,e_8)$ of ratio $G_1$; the red gyrative edges form $\mathfrak g_2 = (e_4,e_6)$ of ratio $G_2$. For simplicity, let all capacitors have capacitance $C$, inductors have inductance $L$, and all gyration ratios are set with respect to reference gyration ratio $G_0 = 1$. Its Lagrangian is:
\begin{equation}\label{eq:l_cwd_a}
    \begin{aligned}
        \mathscr L &= \tfrac{1}{2}(q_1 - q_2)(\dot \phi_2 -\dot \phi_1) - \tfrac{1}{2}(q_1 - q_2)(\dot \phi_4 -\dot \phi_3) + \tfrac{1}{2}(q_3 - q_4)(\dot \phi_6 -\dot \phi_5) - \tfrac{1}{2}(q_3 - q_4)(\dot \phi_8 -\dot \phi_7)\\
        &+ \tfrac{G_1}{2}(\phi_3 - \phi_2)(\dot \phi_8 - \dot \phi_5) - \tfrac{1}{2G_1}(q_1-q_2)(\dot q_4 - \dot q_3)
          + \tfrac{G_2}{2}(\phi_4 - \phi_1)(\dot \phi_7 - \dot \phi_6) - \tfrac{1}{2G_2}(q_2-q_1)(\dot q_3 - \dot q_4)\\
        &- \frac{(q_1 - q_2)^2}{2C} - \frac{(\phi_4 - \phi_3)^2}{2L} - \frac{(q_3 - q_4)^2}{2C} - \frac{(\phi_8 - \phi_7)^2}{2L}\\
    \end{aligned}
\end{equation}
Let us consider taking the dual of only the right sub-circuit. Treat this exclusively as a relabeling transformation on the Lagrangian for now. Therefore the following coordinates become:
\begin{equation}
    \begin{aligned}
        \phi_5 \rightarrow Q_1 \\
        \phi_6 \rightarrow Q_2 \\
        \phi_7 \rightarrow Q_3 \\
        \phi_8 \rightarrow Q_4 \\
        q_3 \rightarrow \Phi_1 \\
        q_4 \rightarrow \Phi_2 \\
    \end{aligned}
\end{equation}
Moreover, the inductor becomes a capacitor with capacitance $L$ and the capacitor becomes an inductor with inductance $C$. Substituting these new coordinates into the Lagrangian yields:
\begin{equation}\label{eq:L1}
    \begin{aligned}
        \mathscr L &= \tfrac{1}{2}(q_1 - q_2)(\dot \phi_2 -\dot \phi_1) - \tfrac{1}{2}(q_1 - q_2)(\dot \phi_4 -\dot \phi_3)
        + \tfrac{1}{2}(\Phi_1 - \Phi_2)(\dot Q_2 -\dot Q_1) - \tfrac{1}{2}(\Phi_1 - \Phi_2)(\dot Q_4 -\dot Q_3)\\[2pt]
        &+ {\tfrac{G_1}{2}(\phi_3 - \phi_2)(\dot Q_4 - \dot Q_1) - \tfrac{1}{2G_1}(q_1-q_2)(\dot \Phi_2 - \dot \Phi_1)}
          + { \tfrac{G_2}{2}(\phi_4 - \phi_1)(\dot Q_3 - \dot Q_2) - \tfrac{1}{2G_2}(q_2-q_1)(\dot \Phi_1 - \dot \Phi_2)} \\
        &- \frac{(q_1 - q_2)^2}{2C} - \frac{(\phi_4 - \phi_3)^2}{2L}
          - {\frac{(\Phi_1 - \Phi_2)^2}{2C}} - {\frac{(Q_4 - Q_3)^2}{2L}}
    \end{aligned}
\end{equation}
Isolating the gyrator terms yields:
\begin{equation}\label{eq:transcwd}
    \begin{aligned}
        \mathscr L_{\mathfrak g_1} &= {\tfrac{G_1}{2}(\phi_3 - \phi_2)(\dot Q_4 - \dot Q_1) - \tfrac{1}{2G_1}(q_1-q_2)(\dot \Phi_2 - \dot \Phi_1)}\\        \mathscr L_{\mathfrak g_2} &=  { \tfrac{G_2}{2}(\phi_4 - \phi_1)(\dot Q_3 - \dot Q_2) - \tfrac{1}{2G_2}(q_2-q_1)(\dot \Phi_1 - \dot \Phi_2)} 
    \end{aligned}
\end{equation}
Varying these equations yields:
\begin{equation}
    \begin{aligned}
        G_1(\dot \phi_3 - \dot \phi_2) &=  (\dot \Phi_2 - \dot \Phi_1)\\
        -G_1(\dot Q_4 - \dot Q_1) &= (\dot q_1-\dot q_2)\\
        G_2 (\dot \phi_4 - \dot \phi_1) &= (\dot \Phi_1 - \dot \Phi_2)\\
        -G_2 (\dot Q_3 - \dot Q_2) &=(\dot q_2-\dot q_1)\\
    \end{aligned}
\end{equation}
which are the constitutive relations for ideal transformers with turns ratios of $G_1/G_0$ and $G_2/G_0$, respectively. Following our convention for the ideal transformer presented in Section \ref{sec:combo}, one may realize the this circuit with gyrators in cascade as shown in Figure \ref{fig:cwd_b}. Here, we name the gyrators $\mathfrak{g}_a = (e_2,x_1)$ and $\mathfrak{g}_b = (x_2,e_8^*)$ with gyration ratios chosen such that $G_a/G_b = G_1$ and $\mathfrak{g}_c = (e_4,y_1)$ and $\mathfrak{g}_d = (y_2,e_6^*)$ with gyration ratios chosen such that $G_c/G_d = G_2$. These circuits have the same EOM and spectra. Then it can be more clearly seen that when taking a component-wise dual, any gyrator with only one gyrative edge on the dualized really-connected component becomes an ideal transformer.

We could have more easily presented the above example with a single gyrator, however, using the transformation in circuits with multiple gyrators will be essential to later arguments. We will use this componentwise duality transformation extensively in Section \ref{sec:fundrecip}.

\subsection{Partial Dual Transformations}\label{sec:rules}
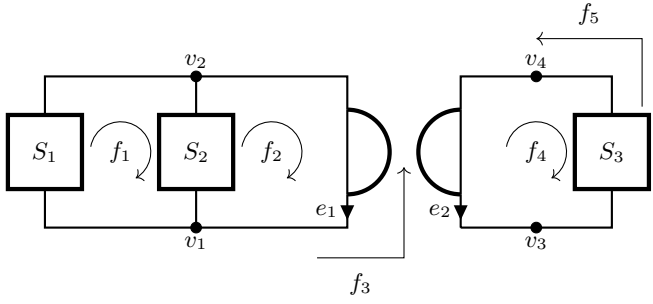
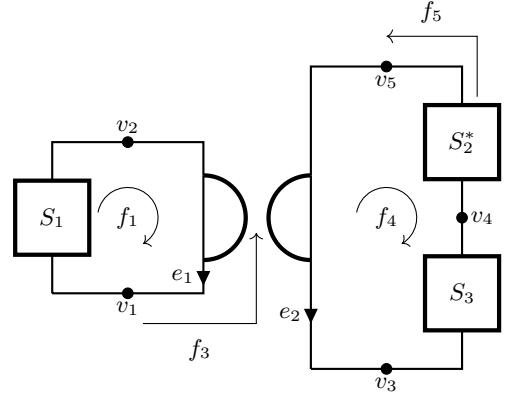
\begin{figure}
  \begin{subfigure}{0.49\linewidth}
        \begin{circuitikz}[scale=2]
        
              \draw[thick] (0,0) to (-1,0) to[twoport,t=$S_1$] (-1,1) to (1,1) 
                           to[half gyrator, i_=$e_1$] (1,0) to (0,0);
        
              \draw[thick] (0,0) to[twoport,t=$S_2$] (0,1); 
              \filldraw[black] (0,0) circle (1pt) node[anchor=north]{$v_1$};
              \filldraw[black] (0,1) circle (1pt) node[anchor=south]{$v_2$};
              \draw[thin, <-] (-0.5,.5) node{$f_1$}  ++(-60:0.2) arc (-60:170:0.2);
              \draw[thin, <-] (0.5,.5) node{$f_2$}  ++(-60:0.2) arc (-60:170:0.2);
              \draw[thin,->] (.8,-.2) to[short,l_=$f_3$] (1.375,-.2) to (1.375,.4);
        
              \draw[thick] (1.75,0) to (2.75,0) to[twoport,t=$S_3$] (2.75,1) to (1.75,1) 
                           to [half gyrator, mirror, i_=$e_2$] (1.75,0);
        
              \filldraw[black] (2.25,0) circle (1pt) node[anchor=north]{$v_3$};
              \filldraw[black] (2.25,1) circle (1pt) node[anchor=south]{$v_4$};
              \draw[thin, <-] (2.25,.5) node{$f_4$}  ++(-60:0.2) arc (-60:170:0.2);
              \draw[thin,->] (2.95,.8) to (2.95,1.25) to[short,l_=$f_5$] (2.25,1.25);
        
        \end{circuitikz}
    \caption{The diagram above clearly has greater than one really-connected component. Here and in discussion regarding the circuit drawn above, we will suppose that sub-circuit $S_2$ is reciprocal. That is to say that sub-circuit $S_2$ is allowed to contain only inductors and capacitors, though we do allow non-linear elements such as Josephson junctions or coherent quantum phase slips. }\label{fig:generic_left}
  \end{subfigure}
  \begin{subfigure}{0.49\linewidth}
    \begin{circuitikz}[scale=2]
      \draw[thick] (0,0) to[twoport,t=$S_1$] (0,1)  to (1,1) to[half gyrator,i_=$e_1$] (1,0) to (0,0);
      \filldraw[black] (0.5,0) circle (1pt) node[anchor=north]{$v_1$};
      \filldraw[black] (0.5,1) circle (1pt) node[anchor=south]{$v_2$};
      \draw[thin, <-] (0.5,.5) node{$f_1$}  ++(-60:0.2) arc (-60:170:0.2);
      \draw[thin,->] (.6,-.2) to[short,l_=$f_3$] (1.35,-.2) to (1.35,.4);

\begin{scope}[xshift=6]
    \draw[thick] (1.5,-.5) to (2.5,-.5) to[twoport,t=$S_3$] (2.5,.5) to[twoport,t=$S_2^*$] (2.5,1.5) to (1.5,1.5) to[half gyrator, mirror,i_=$e_2$] (1.5,-.5);
      \filldraw[black] (2,-.5) circle (1pt) node[anchor=north]{$v_3$};
      \filldraw[black] (2.5,.5) circle (1pt) node[anchor=west]{$v_4$};
      \filldraw[black] (2,1.5) circle (1pt) node[anchor=north]{$v_5$};
      \draw[thin, <-] (2,.5) node{$f_4$}  ++(-60:0.2) arc (-60:170:0.2);
      \draw[thin,->] (2.6,1.3) to (2.6,1.7) to[short,l_=$f_5$] (2,1.7);
\end{scope}

    \end{circuitikz}
    \caption{Another drawing of a circuit with greater than one really-connected component. We also demand that sub-circuit $S_2^*$ is reciprocal. }\label{fig:partialdualb}
  \end{subfigure}
  \caption{The labeled squares denote two-port devices with some finite number of internal nodes and loops. 
  We aim to demonstrate that, so long as $S_2$ is reciprocal, there exists a reciprocal two-port lumped-element sub-circuit $S_2^*$ such that the circuit in (a) has a Hamiltonian related to that of (b) by a canonical transformation}\label{fig:generic}
\end{figure}

We will now show that linear gyrators dualize lumped-element circuits in a manner that is independent of linearity. To this end, consider the drawings in Figure \ref{fig:generic}. We have partitioned a circuit into three sub-circuits $S_1$, $S_2$, and $S_3$, which we leave generic up to the assumption that $S_2$ is reciprocal (it is composed exclusively of reciprocal circuit elements). As discussed in Section \ref{sec:coor}, we are free to independently choose coordinates for individual really-connected components, a freedom we use extensively.
We will express the really-connected component containing $S_1$ and $S_2$ in edge-vertex coordinates and the remaining really-connected component will instead be expressed in face-edge coordinates.

Since there are two really-connected components, there are also gauge freedoms associated with each one. For simplicity, we will gauge $\phi_{v_2} = q_{f_4} = 0$. We can then write the Lagrangian for the total system as 
\begin{equation}
    \mathscr L = 
    \mathscr L_{S_1} + 
    \mathscr L_{S_2} + 
    \mathscr L_{S_3} +
    \frac{1}{2}Q_{e_1}\dot\phi_{v_1} -\frac{1}{2} q_{f_5} \dot\Phi_{e_2} + 
    \frac{1}{2}G_{(e_1,e_2)} \phi_{v_1}\dot\Phi_{e_2} - \frac{1}{2G_{(e_1,e_2)}}Q_{e_1}\dot q_{f_5}
\end{equation}
where $\mathscr L_{S_1}$ and $\mathscr L_{S_2}$ depend only upon $\phi_{v_1}$ and degrees of freedom internal to their sub-circuits. 
Likewise, $\mathscr L_{S_3}$ depends upon only $q_{f_5}$ and internal degrees of freedom. 

Observe that the EOM are 
\begin{subequations}
    \begin{align}
      0 &= \frac{\delta}{\delta \phi_{v_1}} \int \mathrm d t\, \mathscr L_{S_1} +
      \frac{\delta}{\delta \phi_{v_1}} \int \mathrm d t\, \mathscr L_{S_2}  - \frac{1}{2} \dot Q_{e_1} + \frac{G_{(e_1,e_2)}}{2}\dot\Phi_{e_2}. \\ 
      0 &= \dot \phi_{v_1} - \frac{1}{G_{(e_1, e_2)}} \dot q_{f_5} \\ 
      0 &= \frac{\delta}{\delta q_{f_5}} \int \mathrm d t \mathscr L_{S_3}- \frac{1}{2} \dot\Phi_{e_2} + \frac{1}{2 G_{(e_1,e_2)}} \dot Q_{e_1} \\ 
      \end{align}
    \end{subequations}
    The second equation is a constraint rather than a dynamical equation: $Q_{e_1}$ appears in $\LL$ only as a factor multiplying $\dot\phi_{v_1} - \dot q_{f_5}/G_{(e_1,e_2)}$, so it acts as a Lagrange multiplier, and we are free to solve its equation of motion and substitute $\phi_{v_1} \mapsto \frac{1}{G_{(e_1,e_2)}}\big(q_{f_5} + c\big)$, for some fixed $c$, into $\LL_{S_2}$ without changing the remaining equations of motion. Consistency requires the same factor of $G_{(e_1,e_2)}$ to accompany the relabeling transformation on the degrees of freedom internal to $S_2$, which is therefore
\begin{equation}\label{eqn:Gscaling}
\begin{aligned}
    Q_e &\rightarrow G_{(e_1,e_2)}\,\Phi_e \\
    \phi_v &\rightarrow q_v / G_{(e_1,e_2)}
\end{aligned}
\end{equation}
rather than the exchange of Section \ref{sec:transformations}. The factors of $G_{(e_1,e_2)}$ cancel in the symplectic term of $\mathscr L_{S_2}$, while each element energy becomes that of its dual evaluated at a rescaled argument,
\begin{equation}\label{eqn:elementrules}
    E^{I*}_e(\Phi) = E^C_e\big(G_{(e_1,e_2)}\Phi\big),
    \qquad
    E^{C*}_e(Q) = E^I_e\big(Q/G_{(e_1,e_2)}\big),
\end{equation}
or $L^*_e = C_e/G_{(e_1,e_2)}^2$ and $C^*_e = G_{(e_1,e_2)}^2 L_e$ for linear elements. What remains is a circuit Lagrangian recognizable as coming from Figure \ref{fig:partialdualb}, where $S_2^*$ indicates the dual of sub-circuit $S_2$ with its parameters rescaled in this way, and the port carries charge $q_{f_5} + c$. For a linear $S_2$ the constant $c$ does not affect the spectrum. For a non-linear $S_2$ it acts as an offset charge, as in the transmon, and the two circuits are isospectral when $c$ equals the conserved value of $G_{(e_1,e_2)}\phi_{v_1} - q_{f_5}$.

The partial dual transformation is a point transformation and can therefore be undone: a sub-circuit moved into series with $e_2$ may be moved back into parallel with $e_1$. More generally, since duality is an involution, any reciprocal sub-circuit in series with $e_2$ is the dual of one in parallel with $e_1$, and so may itself be moved into parallel with $e_1$. Exchanging the roles of $e_1$ and $e_2$ gives the same statements in the other direction, with the rescaling (Equation \ref{eqn:elementrules}) unaffected since it depends only on $G_{(e_1,e_2)}^2$. Together, these allow a reciprocal sub-circuit in series or in parallel with either gyrative edge of a gyrator to be moved to the other side and back, as in Figure \ref{fig:partialduality} and step 6 of the algorithm in Section \ref{sec:family}.

Through the partial dual transformation, we recover the well-known insight that gyrators can simulate inductors (capacitors) using capacitors (inductors) \cite{tellegen}. There are practical limits to the values of inductances and capacitances that can be reliably fabricated because of nonzero minimum feature size, finite chip size, and other considerations \cite{manucharyan_fluxonium_2009,devoret_leshouches,vool_introduction_2017}. As such, we draw attention to the fact that inductors with large inductance may be simulated with gyrators and capacitors with large capacitance. Likewise, capacitors with small capacitance may be simulated by using inductors with small inductance and a gyrator. The partial dual transformation extends Tellegen's conclusion to non-linear elements, with the element rescaled according to (\ref{eqn:elementrules}). The quantum phase slip, dual to the Josephson junction, has been primarily implemented with superconducting nanowires \cite{mooij_superconducting_2006}. However, Josephson junctions have more mature fabrication procedures and therefore, together with gyrators (which have been implemented in superconducting hardware \cite{KollarNR}), they provide an alternative means to engineer a QPS. In particular a Josephson junction with energy $-E_J\cos(2\pi\Phi/\Phi_0)$ across a gyrator with ratio $G$ is the QPS with energy $-E_J\cos\big(2\pi Q/(G\Phi_0)\big)$, with charge period $G\Phi_0$.

\subsection{Continuous families of partial dual circuits}\label{sec:family}
\begin{figure}
    \centering
    \begin{subfigure}{0.49\linewidth}
    \begin{circuitikz}[scale=2.5]
    \draw[thick] (0,0) to[josephson] (1,0) to[josephson] (1,1) to[capacitor] (0,1) to[inductor] (0,0) to[cI] (-.5,.5) to (0,1) to (.4,1.4) to[inductor](1,1) to[inductor] (1.7,.5) to[capacitor] (1,0) ;  
    \end{circuitikz}
    \caption{} \label{fig:sub_a}
    \end{subfigure}
    \begin{subfigure}{0.49\linewidth}
    \begin{circuitikz}[scale=2.5]
    \draw[thick] (0,0) to[josephson] (1,0) to[josephson] (1,1) to[capacitor] (0,1) to[inductor] (0,0) to[cI] (-.6,0) to[half gyrator] (-.6,1) to (0,1)  to (.4,1.4) to[inductor](1,1) to[inductor] (1.7,.5) to[capacitor] (1,0) ;  
    \draw[thick] (-1.2, 0) to[half gyrator, mirror] (-1.2,1);
    \end{circuitikz}
    \caption{} \label{fig:sub_b}
    \end{subfigure}
    \begin{subfigure}{0.49\linewidth}
    \begin{circuitikz}[scale=2.5]
    \draw[thick] (-0.5,0) to[josephson] (0.5,0) to[josephson] (0.5,1) to[capacitor] (-0.5,1) to[inductor] (-0.5,0) to (-.7,0) to[half gyrator] (-.7,1) to (-0.5,1)  to (-0.1,1.4) to[inductor](0.5,1) to[inductor] (1.2,.5) to[capacitor] (0.5,0) ;  
    \draw[thick] (-1.3, 0) to[half gyrator, mirror] (-1.3,1) to (-1.8,1) to[josephson] (-1.8,0) to (-1.3,0);
    \end{circuitikz}
    \caption{} \label{fig:sub_c}
    \end{subfigure}
    \begin{subfigure}{0.49\linewidth}
    \begin{circuitikz}[scale=2.5]
    \draw[thick] (0.3,0) to[josephson] (1.3,0) to[josephson] (1.3,1) to[capacitor] (0.3,1);
    \draw[thick] (0.3,0) to[half gyrator] (0.3,1);
    \draw[thick] (0.3,1)  to (.7,1.4) to[inductor](1.3,1) to[inductor] (2,.5) to[capacitor] (1.3,0) ;  
    \draw[thick] (-.3, 0) to[half gyrator, mirror] (-.3,1) to (-1.3,1) to[josephson] (-1.3,0) to[capacitor] (-.3,0);
    \end{circuitikz}
    \caption{} \label{fig:sub_d}
    \end{subfigure}
    \begin{subfigure}{0.49\linewidth}
    \begin{circuitikz}[scale=2.5]
    \draw[thick] (0,0) to[half gyrator] (0,1);
    \draw[thick] (0,0) to[josephson] (1,0) to[josephson] (0,1);
    \draw[thick] (0,1) to[inductor] (1,1) to[capacitor] (1,0) ;  
    \draw[thick] (-.6, 0) to[half gyrator, mirror] (-.6,1) to (-1.6,1) to[josephson] (-1.6,.5)  to[capacitor] (-1.6,0) to (-.6,0);
    \draw[thick] (-1.1,0) to[capacitor] (-1.1,.5) to[inductor] (-1.1,1);
    \end{circuitikz}
    \caption{} \label{fig:sub_e}
    \end{subfigure}
    \begin{subfigure}{0.49\linewidth}
    \begin{circuitikz}[scale=2.5]
    \draw[thick] (0.3,0) to[half gyrator] (0.3,1);
    \draw[thick] (0.3,0) to[short] (1.3,0) to[josephson] (0.3,1);
    \draw[thick] (0.3,1) to[inductor] (1.3,1) to[capacitor] (1.3,0) ;  
    \draw[thick] (-.3, 0) to[half gyrator, mirror] (-.3,1) to (-1.5,1) to[josephson] (-1.5,.5)  to[capacitor] (-1.5,0) to (-.3,0);
    \draw[thick] (-1,0) to[capacitor] (-1,.5) to[inductor] (-1,1);
    \draw[thick] (-.7,0) to[cI] (-.7,1);
    \end{circuitikz}
    \caption{} \label{fig:sub_f}
    \end{subfigure}
    \begin{subfigure}{0.49\linewidth}
    \begin{circuitikz}[scale=2.5]
    \draw[thick] (.9,0) to[half gyrator] (.9,1);
    \draw[thick] (.9,1) to[inductor] (1.6,.5) to[capacitor] (.9,0) ;  
    
    \draw[thick] (.3, 0) to[half gyrator, mirror] (.3,1) to[cI] (-.7,1) to (-1.5,1) to[josephson] (-1.5,.5)  to[capacitor] (-1.5,0) to (.3,0);
    \draw[thick] (-1,0) to[capacitor] (-1,.5) to[inductor] (-1,1);
    \draw[thick] (-.7,0) to[cI] (-.7,1);
    \end{circuitikz}
    \caption{} \label{fig:sub_g}
    \end{subfigure}
    \begin{subfigure}{0.49\linewidth}
    \begin{circuitikz}[scale=2.5]
    \draw[thick] (.9,0) to[half gyrator] (.9,1); 
    
    \draw[thick] (.3, 0) to[half gyrator, mirror] (.3,1) to[cI] (-.7,1) to (-1.5,1) to[josephson] (-1.5,.5)  to[capacitor] (-1.5,0) to (-.7,0) to (-.2, .4) to[inductor] (.3,0) to[capacitor] (-.7,0);
    \draw[thick] (-1,0) to[capacitor] (-1,.5) to[inductor] (-1,1);
    \draw[thick] (-.7,0) to[cI] (-.7,1);
    \end{circuitikz}
    \caption{}
    \label{fig:sub_h}
    \end{subfigure}
    \begin{subfigure}{0.49\linewidth}
    \begin{circuitikz}[scale=2.5]
    
    \draw[thick] (.3, 0) to[short] (.3,1) to[cI] (-.7,1) to (-1.5,1) to[josephson] (-1.5,.5)  to[capacitor] (-1.5,0) to (-.7,0) to (-.2, .4) to[inductor] (.3,0) to[capacitor] (-.7,0);
    \draw[thick] (-1,0) to[capacitor] (-1,.5) to[inductor] (-1,1);
    \draw[thick] (-.7,0) to[cI] (-.7,1);
    \end{circuitikz}
    \caption{}
    \label{fig:sub_i}
    \end{subfigure}
    \caption{A large number of circuits that are equivalent for any nonzero finite number $G$. Note that (a) and (i) are dual up to the canonical rescaling (Equation \ref{eqn:Gscaling}): each element of (i) is the $G$-scaled dual of an element of (a). }
    \label{fig:partialduality}
\end{figure}
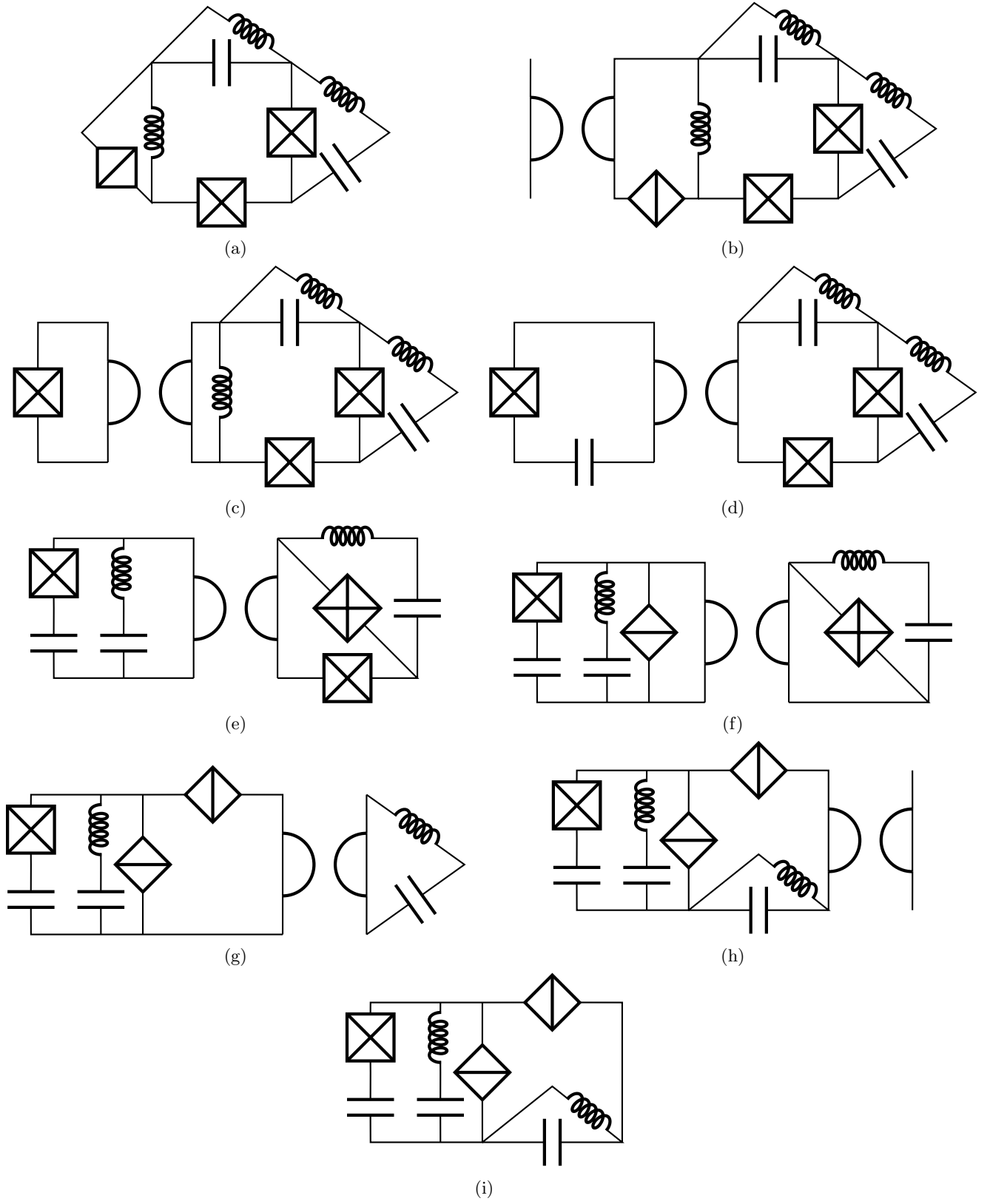
For an arbitrary connected nonsingular reciprocal circuit with Hamiltonian $H$, it is possible to create any number of non-reciprocal circuits that give rise to a Hamiltonian $H'$ where $H'$ and $H$ have the same eigenvalues. 
We present an algorithm for constructing such circuits from any reciprocal circuit. 
\begin{enumerate}
    \item Draw a reciprocal, connected, nonsingular, well-posed LCG circuit, $X$. 
    \item Choose a vertex $v$ of the circuit $X$, and draw an additional three nodes $v'$, $u$, and $u'$. 
    \item Redraw the edges $e$ for which $A_{ev} \neq 0$ in $X$ such that some remain attached to $v$ and others are instead transplanted to be incident upon $v'$
    \item Draw a gyrative edge from $u$ to $u'$ and a gyrative edge from $v$ to $v'$. 
    \item Choose any finite nonzero number $G$, and fix the gyration ratio of $((u,u'),(v,v'))$ to be equal to $G$. 
    \item By using the equivalence transformations of Section \ref{sec:rules}, move any number of sub-circuit elements across $((u,u'),(v,v'))$. 
\end{enumerate}
We remark that if one elects to move \emph{all} of $X$ across the gyrator $((u,u'),(v,v'))$, the resulting circuit has the same form as the dual of $X$ (in the sense that their drawings are the same), but the resulting Lagrangian differs from that of the dual of $X$ by the canonical rescaling (\ref{eqn:Gscaling}), under which each element parameter carries the factors of (\ref{eqn:elementrules}). The rescaling acts reciprocally on the members of each conjugate pair and therefore preserves the symplectic form and the spectrum.
An example of drawings appearing as steps in the algorithm proposed above may be found in Figure \ref{fig:partialduality}. 

We remark that the algorithm posed above may be restarted from step 1 any number of times without altering the EOM. In particular, it is possible to execute the algorithm above with $n > 1$ gyrators separating really-connected components. The problem of counting all possible circuit topologies that may be acquired by this process given some initial circuit $X$ may be a problem of interest as well. We denote the following corollaries of the circuit equivalence transformations as the \textit{gyrator deletion} rules which will play an integral part in determining which Hamiltonians may only be represented with non-reciprocal components.

\begin{figure}
    \centering
    \begin{circuitikz}[scale=2.5]

        \draw[thick] (1, 0) to[inductor, mirror] (1,1) to (0,1); 
        \draw[thick] (.5, 0) to[capacitor, mirror] (.5,1); 
        \draw[thick] (1,0) to (0, 0) to[cI, mirror] (0,1); 

        \draw[{Triangle[length=3mm, width=4mm]}-{Triangle[length=3mm, width=4mm]}, line width=4pt] (1.3, 0.5) -- (2.2, 0.5);

        \begin{scope}[shift={(2, 0)}]
            \draw[thick] (1, 0) to[half gyrator, mirror] (1,1); 
            \draw[thick] (.5, 0) to[josephson, mirror] (.5,1); 
            \draw[thick] (.5,1) to (1,1); 
            \draw[thick] (.5,0) to (1,0); 
            
            \begin{scope}[shift={(1.6, 0)}]
                \draw[thick] (1, 0) to[inductor, mirror] (1,1) to (0,1); 
                \draw[thick] (.5, 0) to[capacitor, mirror] (.5,1); 
                \draw[thick] (1,0) to (0, 0) to[half gyrator] (0,1); 
            \end{scope}
        \end{scope}

    \end{circuitikz}
    \caption{A traditional ill-posed circuit with a QPS without a series inductance and another ill-posed circuit that are related by application of the partial dual transformation to insert a gyrator.}
    \label{fig:illposeddual}
\end{figure}
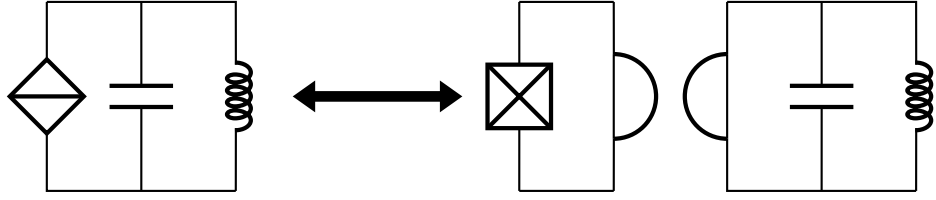

Two further corollaries of the circuit equivalence transformations, \textit{the open-terminated} and \textit{closed-terminated gyrator deletion} rules, play an integral part in determining which Hamiltonians may only be represented with non-reciprocal components.  These corollaries specialize the constitutive relation of a gyrator to the case in which one of its two gyrative edges is constrained, by adjacent circuit topology, to carry no current (a series open) or no voltage (a parallel short). Each constraint propagates through the gyrator to its paired gyrative edge, allowing the entire gyrator to be deleted from the circuit and its other gyrative edge replaced by a trivial element. As the reader will see, one is unlikely to draw a circuit diagram with such gyrators but will find that repeated use of the partial dual transformation will often result in these structures that make a gyrator's contribution trivial to a circuit's dynamics. Thus we use these rules to reduce LCG circuits in Section \ref{sec:fundrecip}.

A gyrative edge in series with an open may be deleted: it carries no current, and by the gyrator constitutive relation its paired gyrative edge then has no voltage across it and becomes a short. This is demonstrated by the transformation from Figure \ref{fig:sub_h} to Figure \ref{fig:sub_i}. Intuitively, a gyrative edge without current through it demands that its corresponding edge has no potential difference across it by the gyrator constitutive relation. Meanwhile, if a gyrative edge is shorted to itself, then it may be deleted and its other half is identically an open. A gyrative edge without potential difference across it will cause its corresponding half to have no current through it.

Finally, we note that this point transformation given by the partial dual preserves the EOM and necessarily preserves well-posedness.  For an ill-posed circuit due to a QPS without a series inductance or a JJ without a parallel capacitance, any circuit that may be related to it by a composition of the partial duality transformation due is itself ill-posed and it will have multi-valued constitutive relations (see Figure \ref{fig:illposeddual}).  In reality, these non-linear superconducting circuit elements intrinsically behave as the linear element required to constrain the solution (e.g. JJ's have their own capacitance and QPS have their own inductance).\footnote{See \cite{rymarz_consistent_2022} for greater discussion}

\section{Reciprocally-Simulable Circuits}\label{sec:fundrecip}

Sections \ref{sec:combo} and \ref{sec:transformations} provide a set of canonical transformations of LCG circuits: duality, the partial-dual transformation, the gyrator series and parallel deletion rules, the open-terminated and closed-terminated deletion rules, and the cascade-to-transformer rule. Each preserves the spectrum of the circuit Hamiltonian. Therefore, we investigate if, for a circuit $C$ with an arbitrary number of gyrators, there exists a means to replicate its dynamics without non-reciprocal elements. Equivalently,  we ask if there exists a circuit of exclusively reciprocal elements whose Hamiltonian is unitarily related to the Hamiltonian of $C$. If so, we call $C$ reciprocally-simulable. Throughout, we consider an ideal transformer realized by a gyrator cascade as a reciprocal element.

With all of our transformations in hand, we introduce the non-reciprocal coupling graph (NCG), a derived multigraph of a circuit.\footnote{Refer to Appendix \ref{sec:appNCG} for all formal definitions and relevant proofs.} For a given LCG circuit, we consider every really-connected component (RCC) a vertex in the NCG, and every gyrator an edge between the RCCs containing its two edges; a gyrator with both edges in the same RCC represents a self-loop. Parallel NCG edges form cycles of length two and self-loops cycles of length one, so a bipartite NCG has no self-loops but may have parallel edges. The NCG records gyrator-mediated coupling between RCCs but contains no information about internal structure\footnote{With the caveat that the degree of an NCG vertex indicates the number of gyrator gyrative edges inside an RCC.} of each sub-circuit: two circuits with the same NCG may differ in internal elements and topology within each RCC. Naturally, a circuit that is made purely of reciprocal elements will have an NCG without any edges.

We convert our canonical transformations into operations on the NCG. As shown in Figure \ref{fig:generic}, a reciprocal sub-circuit in parallel or in series with a gyrative edge can be moved across the gyrator as its dual by the partial dual transformation of Section \ref{sec:rules}. Two operations follow:
\begin{itemize}
    \item Degree Zero Vertices Are Deleted: A degree zero vertex corresponds to an RCC without a gyrator. Its Lagrangian decouples additively and it cannot cause a circuit's non-reciprocal dynamics (Theorem \ref{thm:degree_zero}).
    \item Degree One Vertices Are Deleted: A degree one vertex contains exactly one gyrative edge. Its reciprocal content is moved across the gyrator as a dual sub-circuit in series with the far edge, and the gyrator is then deleted by the open-terminated rule (Theorem \ref{thm:degree_one}). No gyrator or transformer survives the move.
\end{itemize}
Repeated application of these operations deletes every vertex exactly when the NCG is a forest (a graph without cycles), and so every nonsingular, well-posed LCG circuit with a forest NCG is reciprocally-simulable by a circuit of capacitive and inductive elements alone (Theorem \ref{thm:tree_reciprocal}). Figure \ref{fig:LCG_and_NCG_example} provides an example: the repeated application of the degree one rule corresponds to moving every reciprocal element towards RCC 3, leaving a single degree zero vertex which may be deleted.

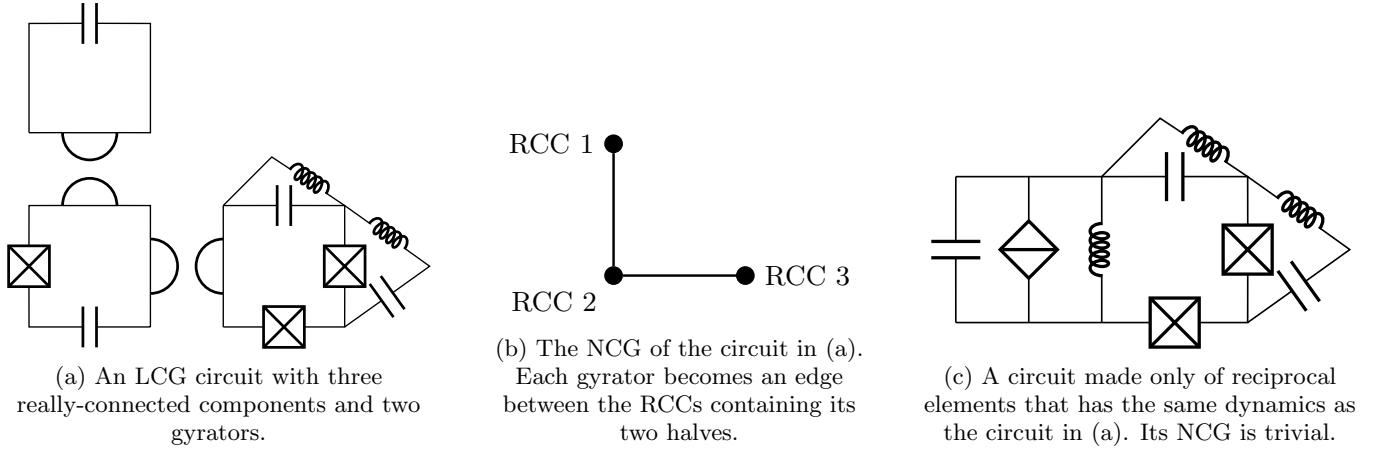
\begin{figure}[htbp]
    \centering
    \begin{subfigure}{.32\linewidth}
        \centering
        \resizebox{\linewidth}{!}{
            \begin{circuitikz}[scale=2.5]
            \draw[thick] (0.3,0) to[josephson] (1.3,0) to [josephson] (1.3,1) to[capacitor] (0.3,1);
            \draw[thick] (0.3,0) to[half gyrator] (0.3,1);
            \draw[thick] (0.3,1)  to (.7,1.4) to[inductor](1.3,1) to[inductor] (2,.5) to[capacitor] (1.3,0) ;  
            \draw[thick] (-.3, 0) to[half gyrator, mirror] (-.3,1) to [half gyrator, mirror](-1.3,1) to[josephson] (-1.3,0) to[capacitor] (-.3,0);
            \draw[thick] (-1.3, 1.6) to[half gyrator, mirror](-.3,1.6) to (-.3,2.5) to [capacitor](-1.3,2.5) to (-1.3, 1.6);
        \end{circuitikz}
        }
        \caption{An LCG circuit with three really-connected components and two gyrators.}
        \label{fig:fig1_panel_circuit}
    \end{subfigure}\hfill
    \begin{subfigure}{.28\linewidth}
        \centering
        \resizebox{\linewidth}{!}{
            \begin{tikzpicture}[scale=1.5,
                vertex/.style={circle, draw=black, fill=black, inner sep=2pt}]
                \node[vertex, label={left:RCC 1}] (v1) at (0,1) {};
                \node[vertex, label={below left:RCC 2}] (v2) at (0,0) {};
                \node[vertex, label={right:RCC 3}] (v3) at (1,0) {};
                \draw[thick] (v1) -- (v2);
                \draw[thick] (v2) -- (v3);
            \end{tikzpicture}
        }
        \caption{The NCG of the circuit in (a). Each gyrator becomes an edge between the RCCs containing its two halves.}
        \label{fig:fig1_panel_ncg}
    \end{subfigure}\hfill
    \begin{subfigure}{.32\linewidth}
        \centering
     \resizebox{\linewidth}{!}{
        \begin{circuitikz}[scale=2.5]
                \draw[thick] (0.3,0) to[josephson] (1.3,0) to[josephson] (1.3,1) to[capacitor] (0.3,1);
                \draw[thick] (0.3,0) to[inductor] (0.3,1);
                \draw[thick] (-0.2,0) to[cI] (-0.2,1);
                \draw[thick] (-0.7,0) to[capacitor] (-0.7,1);
                \draw[thick](-0.7,0) to (.3,0);
                \draw[thick](-0.7,1) to (.3,1);
                \draw[thick] (0.3,1)  to (.7,1.4) to[inductor](1.3,1) to[inductor] (2,.5) to[capacitor] (1.3,0) ;  
           
            \end{circuitikz}
        }
        \caption{A circuit made only of reciprocal elements that has the same dynamics as the circuit in (a). Its NCG is trivial. }
        \label{fig:fig1_panel_blank}
    \end{subfigure}
    
    \caption{LCG Circuit and its NCG. \textbf{(a)} Circuit and its really-connected components (RCCs). \textbf{(b)} NCG for the circuit in (a). \textbf{(c)} The circuit made of evacuating every RCC in (a) towards RCC 3, placing all reciprocal circuit elements there.}
    \label{fig:LCG_and_NCG_example}
\end{figure}

Alternatively, one may consider the class of bipartite NCGs of which forest graphs are a proper subset. If one allows for the use of ideal transformers and treats a cascade of gyrators as a reciprocal element, then any circuit with a bipartite NCG is reciprocally-simulable. Particularly, its spectrum may be realized by a circuit with exclusively capacitive elements, inductive elements, and ideal transformers. For a circuit with a bipartite NCG graph, partition the really-connected components into two sets such that corresponding vertices in the NCG are not adjacent to vertices in its own partition. Apply the component-wise duality transformation (Section \ref{sec:cwd}) to all of the really-connected components in one partition. As all the gyrators in this circuit will connect distinct RCCs (as otherwise the NCG would have a self-loop and fail to be bipartite), then all of these gyrators will be converted to transformers thereby producing a reciprocal circuit.\footnote{See Theorem \ref{thm:bipartite} in the Appendix for the proof}

\section{Conclusion}
This work continues our previous efforts \cite{fluxcharge,osborne2023symplectic} towards a general flux-charge symmetric theory of circuit mechanics. The formalism guarantees the existence of a symplectic reduction for any nonsingular, well-posed LCG circuit with arbitrarily non-linear reciprocal one-port elements.
The circuit's topology is manifest in its Lagrangian, and we provide a taxonomy of its constraints, separating null-vector constraints into pure and gyrator-mixing categories and distinguishing them from those generated by continuous symmetries. As we maintain the flux-charge symmetry of the formalism, we extend the circuit duality point transformation introduced in \cite{fluxcharge} to LCG circuits. Additionally, our Lagrangian formalism allows us to introduce families of isospectral circuits using the partial dual transformation, which acts directly on the circuit Lagrangian. With the partial-dual transformation, we extend Tellegen's classic result \cite{tellegen} to non-linear elements such as the Josephson junction and quantum phase slip. Finally, our gyrator reduction rules act as canonical transformations on the circuit Lagrangian, and their composition provides a sufficient graph-theoretic condition for reciprocal simulability by the examination of the non-reciprocal coupling graph. We prove that any circuit whose NCG is a forest graph is reciprocally-simulable and its Hamiltonian admits a unitarily equivalent circuit realization without gyrators. With the addition of ideal transformers, any circuit with a bipartite NCG is reciprocally-simulable. Practically, the classification of Section \ref{sec:fundrecip} sharpens the role of gyrators in circuit design: gyrators in a reciprocally-simulable circuit are a representational choice, while gyrators in a circuit that we cannot reduce may be dynamically necessary. Identifying simple circuits whose dynamics require non-reciprocal elements could inform the realization of complex non-reciprocal elements for quantum information processing.

\section{Acknowledgements}
We acknowledge useful discussions with Basil Smitham, J. Wilson Staples, and Siddhant Midha. We acknowledge support from the Co-design Center for Quantum Advantage (C$^2$QA) under Contract No. DESC0012704.

\section{Ethics Declarations}
\textbf{Declaration of Generative AI Assistance:} During the preparation of this work, the authors used Anthropic's Claude Fable 5.1 \cite{anthropic2026claude} to assist in the proof of Theorem \ref{thm:bipartite}. The authors independently evaluated, restructured, and verified all mathematical steps, and take full responsibility for the correctness of the final proof. 

\textbf{Competing Interests:} Princeton University Professor Andrew Houck was also a consultant for Quantum Circuits Incorporated (QCI), now D-Wave. Due to his income from QCI, Princeton University has a management plan in place to mitigate a potential conflict of interest that could affect the design, conduct, and reporting of this research.

\section{Software Release}
An accompanying Python package, \texttt{fluxcharge}, for designing LCG circuits within our formalism is available at \url{https://github.com/cgs2155/flux-charge}. It is developed and maintained by Ciro Salcedo. Bug reports and questions can be submitted via the project's issue tracker.

\begin{appendix}

\section{Formalizing Reciprocally-Simulable Circuits}\label{sec:appNCG}
With the canonical transformations provided by the gyrator deletion rules and partial duality, we may begin to classify which circuit's dynamics may be replicated without non-reciprocal elements. 

First we formalize the ideal transformer as written in Section \ref{sec:combo} to classify it as a reciprocal element:
\begin{defn}[\textbf{Cascade Pair}]\label{defn:cascadepair}
    Let $C$ be an LCG circuit. Two gyrators of $C$ form a \textbf{cascade pair} if one gyrative edge of each shares both endpoints with one edge of the other and those endpoints are incident on no other edge.
\end{defn}
\begin{defn}[\textbf{Reciprocally-Simulable Circuit}]\label{defn:recip_sim}
    A circuit $C$ is said to be \textbf{reciprocally-simulable} if there exists a circuit $C'$ without non-reciprocal elements whose Hamiltonian is unitarily related to that of $C$. An ideal transformer, realized as a cascade pair (Definition \ref{defn:cascadepair}), is counted as a reciprocal element; equivalently, $C$ is reciprocally-simulable if there exists such a $C'$ in which every gyrator belongs to a cascade pair.
\end{defn}

We remark that a point transformation of a Lagrangian leaves the EL equation invariant. The Legendre transform sends a point transformation to a canonical transformation of the Hamiltonian and, following quantization, is a unitary transformation on the Hilbert space. Much of the following section involves transforming a circuit $C$ to $C'$ using compositions of point transformations. These operations are therefore spectrum-preserving and the Hamiltonians of both circuits will be unitarily equivalent. Our theorems will use proofs made on the Lagrangian level, but conclusions about the Hamiltonian will follow automatically. To probe the reciprocal simulability of circuits we define the following:

\begin{defn}[\textbf{Reciprocal Block}]\label{defn:recip_block}
    Let $C = (\mathcal V, \mathcal E, \mathcal F, \mathcal C, \mathcal I, \mathcal G, \fG)$ be a circuit, and let $G_{\mathrm{rec}}$ be the subgraph obtained from the circuit graph by keeping only the reciprocal edges $\mathcal I \cup \mathcal C$ and deleting every gyrative edge. A \textbf{reciprocal block} of $C$ is a connected component of $G_{\mathrm{rec}}$ containing at least one edge, viewed as the sub-circuit
    \begin{equation}
        S = \big(\mathcal V_S,\, \mathcal E_S,\, \mathcal F_S,\, \mathcal C \cap \mathcal E_S,\, \mathcal I \cap \mathcal E_S,\, \emptyset,\, \emptyset\big),
    \end{equation}
    where $\mathcal V_S$ is the component's vertex set,
    \begin{equation}
        \mathcal E_S = \{\, e \in \mathcal I \cup \mathcal C : A_{ev} \neq 0 \Rightarrow v \in \mathcal V_S \,\}
    \end{equation}
    its reciprocal edges, and $\mathcal F_S$ the loop set of $S$ under the convention of Section \ref{sec:notation}. By construction a reciprocal block contains only reciprocal elements, $\mathcal G \cap \mathcal E_S = \emptyset$.
\end{defn}

Because $G_{\mathrm{rec}}$ is obtained by deleting edges from the circuit graph, each of its connected components lies inside a single connected component of the full graph: every reciprocal block is therefore contained in a unique RCC. A single RCC may contain many reciprocal blocks, joined to one another (and to other RCCs) only through gyrative edges, which $G_{\mathrm{rec}}$ omits. Distinct reciprocal blocks share no vertices, and every reciprocal edge of $C$ lies in exactly one block.

\begin{defn}[\textbf{Non-Reciprocal Coupling Graph}]\label{defn:NCG}
    Let $C$ be an LCG circuit. The \textbf{non-reciprocal coupling graph} (NCG) of $C$ is the multigraph $G_N = (\mathcal{V}', \mathcal{E}')$ defined by:
    \begin{enumerate}
        \item $\mathcal{V}'$ is the set of really-connected components of $C$.
        \item For each gyrator $\mathfrak{g} = (e_1, e_2) \in \fG$, $\mathcal{E}'$ contains one edge $\varepsilon_\mathfrak{g}$ whose endpoints are the RCCs containing $e_1$ and $e_2$. If these coincide, $\varepsilon_\mathfrak{g}$ is a self-loop.
    \end{enumerate}
    The degree of an NCG vertex counts incidences, so a self-loop contributes $2$. Here, a self-loop is a $1$-cycle and a pair of parallel edges a $2$-cycle.
\end{defn}

The NCG is in general a multigraph: distinct gyrators between the same pair of RCCs contribute distinct parallel edges, and a gyrator whose edge pair lies entirely within a single RCC contributes a self-loop. Every gyrative edge of $C$ has both its endpoints in a single RCC, since RCCs are connected components of the circuit graph using all edge types. A reciprocal block is a connected reciprocal sub-circuit of $C$, finer than an RCC as an RCC may contain multiple reciprocal blocks separated by gyrative edges. Reciprocal blocks are the natural objects to which the partial duality construction is applied; they need not coincide with NCG vertices. 

\begin{defn}[\textbf{Series and Parallel Sub-Circuits}]\label{defn:seriesparallel}
    Let $e = (v, v')$ be a gyrative edge and let $S$ be a sub-circuit of reciprocal elements. $S$ is \textbf{in parallel with} $e$ if every vertex of $S$ other than $v$ and $v'$ is incident only on edges of $S$, so that $S$ shares no vertex with the rest of the circuit other than the endpoints of $e$ and $\mathscr L_S$ depends on the rest of the circuit only through $\Phi_e$. We say $S$ is \textbf{in series with} $e$ if every edge incident on $v$ other than $e$ belongs to $S$, and $S$ shares only vertex $v$ and one other vertex with the rest of the circuit, so that $\mathscr L_S$ depends on the rest of the circuit only through $Q_e$.
\end{defn}

With the partial dual transformation of Section \ref{sec:rules} and the gyrator deletion rules we may now perform operations exclusively on the NCG. To this end we define two operations to reduce it:

\begin{thm}[\textbf{Degree-zero deletion}]\label{thm:degree_zero}
    Let $C$ be a nonsingular, well-posed LCG circuit and let $G_N$ be its NCG. Suppose $v \in \mathcal{V}'$ has degree zero in $G_N$, with corresponding really-connected component $R$. Let $C \setminus R$ denote the circuit obtained by deleting $R$ from $C$. Then the Lagrangian splits as
    \begin{equation}\label{eqn:degree_zero_split}
        \mathscr L_C = \mathscr L_R + \mathscr L_{C \setminus R},
    \end{equation}
    so $C$ is reciprocally-simulable if $C \setminus R$ is.
\end{thm}

\begin{proof}
    By the definition of an RCC, $R$ shares no vertices or edges with $C \setminus R$: the two sub-circuits are graph-disjoint. By the definition of the NCG, every gyrator with at least one edge in $R$ contributes an NCG edge incident on $v$, so $\deg(v) = 0$ implies no gyrator has any edge in $R$. Hence $R$ contains no gyrators.

    The connection matrix $M$ is block-diagonal over really-connected components by (\ref{eqn:bmapdecomp}). Thus, the matrix element $M_{lv'} = \tfrac{1}{2} \sum_{e} B_{le} (\mathbb{P}_\mathcal{C} - \mathbb{P}_\mathcal{I})_{ee} A_{ev'}$ vanishes whenever $l$ is a loop of one sub-circuit and $v'$ a vertex of the other, as they share no edges. The gyrator term contains no contributions involving $R$, since $R$ has no gyrators. Each energy term is local to a single edge, hence to one sub-circuit. The Lagrangian therefore splits as in (\ref{eqn:degree_zero_split}).

    Since $\mathscr L_C$ splits as a sum over disjoint dynamical variables, the Hamiltonian satisfies $H_C = H_R + H_{C \setminus R}$ on disjoint phase spaces. Since $R$ contains no gyrators, $\mathscr L_R$ is the Lagrangian of a reciprocal circuit. If $C \setminus R$ is reciprocally-simulable, then $\mathscr L_C$ is reciprocal up to a point transformation acting on the $C \setminus R$ factor, so $C$ is reciprocally-simulable. 
\end{proof}
For the purpose of using the NCG to keep track of a circuit's reciprocity, we delete degree-zero vertices from the NCG as their corresponding RCC can never cause a circuit to be non-reciprocal.
\begin{thm}[\textbf{Degree-one deletion}]\label{thm:degree_one}
    Let $C$ be a nonsingular, well-posed LCG circuit and let $G_N$ be its NCG. Suppose $v \in \mathcal{V}'$ has degree one in $G_N$, with unique incident edge $\varepsilon$ corresponding to a gyrator $\mathfrak{g}_0$. Then there exists a nonsingular, well-posed LCG circuit $\widetilde C$, related to $C$ by a point transformation, whose NCG is obtained from $G_N$ by deleting $v$ and $\varepsilon$. If $\widetilde C$ is reciprocally-simulable, then $C$ is reciprocally-simulable. In particular, $\widetilde C$ has exactly one fewer gyrator than $C$, and the deletion introduces no transformer.
\end{thm}
\begin{proof}
    Let $R$ denote the RCC at $v$. Since $\deg(v) = 1$, $\varepsilon$ is not a self-loop, so $\mathfrak{g}_0$ has exactly one gyrative edge $e_1 = (v_1, v_2)$ in $R$, and $\mathfrak{g}_0$ is the only gyrator with a gyrative edge in $R$, since any other would contribute a second NCG edge incident on $v$. Let $u$ be the NCG neighbor of $v$ adjoined by $\varepsilon$, and let $e_2$ be the other gyrative edge of $\mathfrak{g}_0$, which lies in the RCC at $u$.

    Every reciprocal element of $R$ lies in some connected component of $R \setminus e_1$, and since $R$ is connected each such component contains $v_1$ or $v_2$. A component containing both touches the rest of the circuit only at $v_1$ and $v_2$, so it is in parallel with $e_1$, and by the partial dual transformation of Section \ref{sec:rules} it may be turned into its dual (with element parameters adjusted as in \ref{eqn:elementrules}) and moved to the RCC at $u$, in series with $e_2$; the transformation leaves every other element and gyrator of $C$ untouched. Doing so for every such component evacuates $R$ of all reciprocal content in parallel with $e_1$, yielding a circuit $C^{(1)}$, related to $C$ by a point transformation, in which no component of $R \setminus e_1$ contains both endpoints of $e_1$. Then $e_1$ carries no current, since any current through it would have to return through a component of $R \setminus e_1$ containing both of its endpoints; $e_1$ is in series with an open, and the open-terminated gyrator deletion rule of Section \ref{sec:rules} applies: $\mathfrak{g}_0$ is deleted from $C^{(1)}$ and $e_2$ is replaced by a short. Call the resulting circuit $\widetilde C$; its Lagrangian is related to that of $C^{(1)}$ by an algebraic elimination of the constrained gyrator variables, and therefore to that of $C$ by composition.

    In $\widetilde C$, the dual sub-circuits lie in the RCC at $u$, and the short replacing $e_2$ joins two vertices of that RCC. What remains of $R$ consists of the components of $R \setminus e_1$ that contained only one endpoint of $e_1$; they contain no gyrative edge, are now disconnected from the rest of the circuit, and are deleted by Theorem \ref{thm:degree_zero}. The original $v$ therefore no longer corresponds to any RCC, and the NCG of $\widetilde C$ is $G_N$ with $v$ and $\varepsilon$ removed; no other NCG edge is affected, since no other gyrator was incident on $R$. Finally, $\widetilde C$ has exactly one fewer gyrator than $C$, and no transformer has been introduced.
\end{proof}

Simply put, the above rule states that one may evacuate all the reciprocal components off of an RCC with a single gyrative edge. Without the reciprocal components the gyrative edge must either be a short or an open; either way it is deleted and thus that gyrator cannot be responsible for a circuit's NCG's irreducibility. Note that the partial dual transformation moves only sub-circuits of capacitors and inductors that are in series or in parallel with a single gyrative edge. It does not move a cascade pair, since a gyrative edge carries a coupling to its partner in another RCC. Nor does it move reciprocal blocks across a gyrator which is contained within a single RCC. Finally, it cannot move reciprocal blocks that are adjacent to but not in series or in parallel with a gyrative edge (a case that does not happen for sub-circuits corresponding to degree one vertices in an NCG).

An NCG is said to be reducible if repeated application of the degree-zero and degree-one rules (Theorems \ref{thm:degree_zero} and \ref{thm:degree_one}) reduces $G_N$ to the empty graph. Every vertex on a cycle of $G_N$ has degree at least two and keeps it under deletion of other vertices, so an NCG is reducible exactly when it is a forest, and a circuit with a reducible NCG is reciprocally-simulable by a circuit containing no gyrators at all (Theorem \ref{thm:tree_reciprocal}).

\begin{thm}[\textbf{Forest NCGs are reciprocally-simulable}]\label{thm:tree_reciprocal}
    Let $C$ be a nonsingular, well-posed LCG circuit. If the NCG $G_N$ of $C$ is a forest, then $C$ is reciprocally-simulable by a circuit containing no gyrators.
\end{thm}

\begin{proof}
    We proceed by induction on the number of edges of $G_N$.

   \textbf{Base case ($G_N$ has no edges):} Every vertex of $G_N$ has degree zero, so by Theorem \ref{thm:degree_zero} each RCC may be deleted in turn, leaving the empty circuit, which is trivially reciprocal. Hence $C$ is reciprocally-simulable.

    \textbf{Inductive step:} Suppose the claim holds for all nonsingular, well-posed LCG circuits whose NCG is a forest with fewer than $m$ edges, and let $C$ have NCG $G_N$ a forest with $m \geq 1$ edges. Every forest with at least one edge has a leaf (a degree one vertex). Let $v$ be such a leaf, with unique incident edge $\varepsilon$.

    By Theorem \ref{thm:degree_one}, there exists a nonsingular, well-posed LCG circuit $\widetilde C$ related to $C$ by a point transformation, whose NCG is $G_N$ with $v$ and $\varepsilon$ removed. The graph $G_N \setminus \{v, \varepsilon\}$ is a forest on $m - 1$ edges. By the inductive hypothesis, $\widetilde C$ is reciprocally-simulable. Since $\widetilde C$ has Hamiltonian unitarily equivalent to $C$, $C$ is reciprocally-simulable as well. Since Theorem \ref{thm:degree_one} deletes one gyrator at each step and introduces none, the final circuit contains no gyrators.
\end{proof}

For bipartite NCGs that are not forests, we do not rely on the gyrator deletion rules. Instead, the component-wise duality transformation of Section \ref{sec:cwd} converts every gyrator at once into a transformer.

\begin{thm}[\textbf{Bipartite NCGs are reciprocally-simulable}]\label{thm:bipartite}
    Let $C$ be a nonsingular, well-posed LCG circuit and let $G_N$ be its NCG. If $G_N$ is bipartite, then $C$ is reciprocally-simulable.
\end{thm}
\begin{proof}
    Since $G_N$ is bipartite it has no self-loops, and its vertices may be two-colored so that every gyrator $\mathfrak{g} = (e_1, e_2)$ has $e_1$ in an RCC of the first color and $e_2$ in an RCC of the second, reordering the pair if necessary using $G_{(e_2,e_1)} = -G_{(e_1,e_2)}$.

    Apply the component-wise duality transformation of Section \ref{sec:cwd} to the RCCs of the second color. By (\ref{eqn:bmapdecomp}) the Lagrangian (\ref{eqn:no_omega_lagrangian}) splits as $\mathscr L = \sum_i \big[\langle q^{(i)}|M^{(i)}|\dot\phi^{(i)}\rangle - E^{(i)}\big] + \sum_{\mathfrak{g}\in\mathfrak{G}}\mathscr L_{\mathfrak{g}}$, where $E^{(i)}$ collects the element energies of the $i$-th RCC. On each transformed RCC the bracket becomes, up to a total time derivative, the Lagrangian of the planar dual $G^{(i)*}$ with capacitors and inductors exchanged, by (\ref{eqn:Mstar}), (\ref{eqn:ECstar}), (\ref{eqn:EIstar}) and (\ref{eqn:Lstar}); $G^{(i)*}$ is planar because $G^{(i)}$ is. Untransformed RCCs are unaffected. The transformation is a relabeling of coordinates, hence a point transformation, and the EL equations are unchanged.

    Every gyrator has $e_1$ in an untransformed RCC and $e_2$ in a transformed one. Under the relabeling, $\Phi_{e_2} = \sum_v A_{e_2 v}\phi_v$ becomes $\sum_v A_{e_2 v} q^*_v = Q^*_{e_2}$, the edge charge of $e_2^*$ in $G^{(i)*}$, since $B^{(i)*} = (A^{(i)})^\intercal$, and $Q_{e_2} = \sum_l B_{l e_2} q_l$ becomes $\sum_l B_{l e_2}\phi^*_l = \Phi^*_{e_2}$, since $A^{(i)*} = (B^{(i)})^\intercal$. The constitutive relations (\ref{eqn:gyrconst}) therefore read
    \begin{equation}
        \begin{aligned}
                    G_{\mathfrak{g}}\dot\Phi_{e_1} &= \dot\Phi^*_{e_2}\\
                    \dot Q_{e_1} =& -G_{\mathfrak{g}}\,\dot Q^*_{e_2}\\
        \end{aligned}
    \end{equation}
    which are the constitutive relations of an ideal transformer between $e_1$ and $e_2^*$ with turns ratio $G_{\mathfrak{g}}/G_0$; compare Section \ref{sec:combo}, where the cascade of ratios $G_1, G_2$ gave $\dot\Phi_{e_2'} = (G_1/G_2)\dot\Phi_{e_1}$ and $\dot Q_{e_1} = -(G_1/G_2)\dot Q_{e_2'}$. Each such coupling is realized in the LCG element set by a cascade pair (Definition \ref{defn:cascadepair}) of ratios $G_1/G_2 = G_{\mathfrak{g}}/G_0$, whose two inner gyrative edges form an RCC by themselves; the algebraic elimination of that RCC's variables is a point transformation returning the relations above. The resulting circuit $\widetilde C$ is an LCG circuit in which every gyrator belongs to a cascade pair, related to $C$ by a composition of point transformations, and is nonsingular and well-posed because its EL equations are those of $C$. Therefore, $C$ is reciprocally-simulable.
\end{proof}

\end{appendix}

\bibliography{bibliography}
\end{document}